\documentclass[letterpaper,11pt]{article}
\usepackage[cmex10]{amsmath}
\usepackage{amssymb,mathrsfs,amsthm}
\usepackage{float}
\usepackage{framed}
\usepackage{enumerate}
\usepackage{color}
\usepackage{xspace}
\usepackage[normalem]{ulem}
\usepackage[margin=1in]{geometry}
\usepackage{graphicx}

\newtheorem{theorem}{Theorem}[section]

\newtheorem{definition}[theorem]{Definition}
\newtheorem{lemma}[theorem]{Lemma}

\newtheorem{corollary}[theorem]{Corollary}

\newenvironment{protocol}{\begin{framed}}{\vspace{-2ex}\end{framed}\vspace{-2ex}}

\renewcommand{\H}{{\cal H}}
\newcommand{\D}{{\cal D}}
\newcommand{\ket}[1]{| #1 \rangle}
\newcommand{\bra}[1]{\langle #1 |}

\newcommand{\ketbra}[2]{|#1\rangle \!\langle #2|}
\newcommand{\proj}[1]{\ketbra{#1}{#1}}
\newcommand{\zero}{\ket{0}}
\newcommand{\one}{\ket{1}}
\newcommand{\id}{\Bbb{I}}
\newcommand{\tr}{\mathrm{tr}}
\newcommand{\I}{\mathbb{I}}
\newcommand{\dist}{\delta}
\renewcommand{\o}{\otimes}

\newcommand{\set}[1]{\left \{#1\right \}}
\newcommand{\Set}[2]{\left\{ #1 :\, #2\right\}}
\newcommand{\C}{\Bbb{C}}
\newcommand{\R}{\Bbb{R}}

\renewcommand{\P}{P}
\DeclareMathOperator{\Hone}{H}

\DeclareMathOperator{\h}{h}

\newcommand{\ch}{ch}

\newcommand{\key}{K}
\newcommand{\idealkey}{\tilde{K}}
\DeclareMathOperator{\Hull}{Hull}

\newcommand{\eps}{\varepsilon}

\renewcommand{\d}{d}

\newcommand{\user}{{\cal U}}
\newcommand{\PC}{{\sigma}}  

\newcommand{\dP}{\hat{P}}   

\newcommand{\SP}{{\sf PV}}
\newcommand{\basicSP}{\SP_{\text{\rm\tiny \!BB84}}}
\newcommand{\pos}{\text{$po\hspace{-0.1ex}s$}}
\newcommand{\apos}{\text{$apo\hspace{-0.1ex}s$}}
\newcommand{\AUTH}{{\sf AUTH}}
\newcommand{\wAUTH}{{\sf wAUTH}}
\newcommand{\code}{N}
\newcommand{\Cbrl}{C_{\mbox{\tiny\rm $\ell$-BR}}}
\newcommand{\cl}{cl}

\newcommand{\SPch}{{\sf Chlg}}
\newcommand{\SPresp}{{\sf Resp}}
\newcommand{\SPver}{{\sf Ver}}

\newcommand{\url}[1]{{\tt #1}}

\newcommand{\NPE}{No-PE\xspace}

\newcommand{\note}[1]{}

\begin{document}

\title{%
\vspace{-2cm}  
Position-Based Quantum Cryptography: \\ Impossibility and Constructions \ \\ \ }
\author{
\ \ \ \ \ \ \ \ \ \
Harry Buhrman\thanks{Centrum Wiskunde \& Informatica (CWI) and
  University of Amsterdam, The Netherlands. Email: \texttt{Harry.Buhrman@cwi.nl}. Supported
by a NWO VICI grant and the EU 7th framework grant QCS.} \and
Nishanth Chandran\thanks{Department of Computer Science, UCLA, Los Angeles, CA, USA. Email: \texttt{\{nishanth, gelles\}@cs.ucla.edu}.
Supported in part by
NSF grants 0716835, 0716389, 0830803, and 0916574.} \and
Serge Fehr\thanks{Centrum Wiskunde \& Informatica (CWI), Amsterdam, The Netherlands.
Email: \texttt{Serge.Fehr@cwi.nl}.}  \and
Ran Gelles$^\dag$\and
\ \ \ \ \ \ \ \ \ \
Vipul Goyal\thanks{Microsoft Research,
Bangalore, India.
Email: \texttt{vipul@microsoft.com}.} \and
Rafail Ostrovsky\thanks{Department of Computer Science and Mathematics, UCLA,
Los Angeles, CA, USA.
Email: \texttt{rafail@cs.ucla.edu}.
Supported in part by
IBM Faculty Award, Xerox Innovation Group Award,
the Okawa Foundation Award, Intel, Teradata, BSF grant 2008411,
NSF grants 0716835, 0716389, 0830803, 0916574 and U.C.~MICRO grant.} \and
Christian Schaffner\thanks{University of Amsterdam and Centrum
  Wiskunde \& Informatica (CWI), Amsterdam, The Netherlands. Email:
  \texttt{c.schaffner@uva.nl}. Supported by a NWO VENI grant. } \and
}
\date{}
\maketitle

\begin{abstract}
  In this work, we study position-based cryptography in the quantum
  setting.  The aim is to use the geographical position of a party as
  its only credential.  On the negative side, we show that if
  adversaries are allowed to share an arbitrarily large entangled
  quantum state, no secure position-verification is possible at all.
  To this end, we prove the following very general result.  Assume
  that Alice and Bob hold respectively subsystems $A$ and $B$ of a
  (possibly) unknown quantum state $\ket{\psi} \in \H_A \otimes \H_B$.
  Their goal is to calculate and share a new state $\ket{\varphi} =
  U\ket{\psi}$, where $U$ is a fixed unitary operation.  The question
  that we ask is how many rounds of mutual communication are
  needed. It is easy to achieve such a task using two rounds of
  classical communication, whereas in general, it is impossible
  with no communication at all.

  Surprisingly, in case Alice and Bob share enough entanglement to
  start with and we allow an arbitrarily small failure probability, we
  show that the same task can be done using a {\em single} round of
  classical communication in which Alice and Bob simultaneously
  exchange two classical messages.  Actually, we prove that a relaxed
  version of the task can be done with {\em no} communication at all,
  where the task is to compute instead a state $\ket{\varphi'}$ that
  coincides with $\ket{\varphi} = U\ket{\psi}$ up to local operations
  on $A$ and on $B$, which are determined by classical information
  held by Alice and Bob. The one-round scheme for the original task
  then follows as a simple corollary.  We also show that these results
  generalize to more players.  As a consequence, we show a generic
  attack that breaks any position-verification scheme.

  On the positive side, we show that if adversaries do not share any
  entangled quantum state but can compute arbitrary quantum
  operations, secure position-verification is achievable.  Jointly,
  these results suggest the interesting question whether secure
  position-verification is possible in case of a bounded amount of
  entanglement. Our positive result can be interpreted as resolving
  this question in the simplest case, where the bound is set to zero.

  In models where secure positioning is achievable, it has a number of
  interesting applications. For example, it enables secure
  communication over an insecure channel without having any pre-shared
  key, with the guarantee that only a party at a specific location can
  learn the content of the conversation. More generally, we show that
  in settings where secure position-verification is achievable, other
  position-based cryptographic schemes are possible as well, such as
  secure position-based authentication and position-based key
  agreement.
\end{abstract}


\section{Introduction}

\subsection{Background}
The goal
of {\em position-based cryptography} is to use the geographical position
of a party as its only ``credential".
For example, one would like to send a
message to a party at a geographical position~$\pos$ with the
guarantee that the party can decrypt the message only if he or she
is physically present at~$\pos$.
The general concept of position-based cryptography was introduced by Chandran, Goyal, Moriarty and Ostrovsky~\cite{CGMO09}; certain specific related tasks have been considered before under different names (see below and Section~\ref{sec:RelatedWork}).

A central task in position-based cryptography is the problem of {\em
  position-verification}.  We have a {\em prover}~$P$ at position~$\pos$,
  wishing to convince a set of {\em verifiers} $V_0,\ldots,V_k$
(at different points in geographical space) that~$P$ is indeed at that
position~$\pos$. The prover can run an interactive protocol with the
verifiers in order to convince them.  The main technique for such a
protocol is known as distance bounding~\cite{BC93}. In this technique,
a verifier sends a random nonce to~$P$ and measures the time taken for
$P$ to reply back with this value.  Assuming that the speed of
communication is bounded by the speed of light, this technique gives
an upper bound on the distance of~$P$ from the verifier.

The problem of secure positioning has been studied before in the
field of wireless security, and there have been several proposals
for this task (\cite{BC93,SSW,VN04,B04,CH05,SP05,ZLFW06,CCS06}).
However, \cite{CGMO09} shows that there exists no
protocol for secure positioning that offers security in the presence
of {\em multiple colluding} adversaries. In other words, the set of
verifiers cannot distinguish between the case when they are
interacting with an honest prover at~$\pos$ and the case when they
are interacting with multiple colluding dishonest provers, none of
which is at position~$\pos$. Their impossibility result holds even
if one makes computational hardness assumptions, and it also rules out
most other interesting position-based cryptographic tasks.

In light of the strong impossibility result, \cite{CGMO09} considers a setting that assumes restrictions on the parties' storage capabilities, called the Bounded-Retrieval Model (BRM) in the full version of \cite{CGMO09}, and constructs
secure protocols for 
position-verification and for position-based key exchange (wherein the
verifiers, in addition to verifying the position claim of a prover,
also exchange a secret key with the prover). While these protocols
give us a way to realize position-based cryptography, the underlying setting is relatively hard to justify in practice. 

This leaves us with the question: are there any other
assumptions or settings in which position-based cryptography is
realizable?

\subsection{Our Approach and Our Results}

In this work, we study position-based cryptography in the {\em
  quantum} setting. To start with, let us briefly explain why moving
to the quantum setting might be useful. The impossibility result
of~\cite{CGMO09} relies heavily on the fact that an adversary can
locally store all information she receives {\em and} at the same time
share this information with other colluding adversaries, located
elsewhere.  Recall that the positive result of~\cite{CGMO09} in the
BRM circumvents the impossibility result by assuming that an adversary
{\em cannot} store all information he receives.  By considering the
quantum setting, one may be able to circumvent the impossibility
result thanks to the following observation. If some information is
encoded into a quantum state, then the above attack fails due to the
no-cloning principle: the adversary can either store the quantum state
or send it to a colluding adversary (or do something in-between, like
store part of it), but \emph{not both}.

However, this intuition turns out to be not completely accurate. Once
the adversaries pre-share entangled states, they can make use of
quantum teleportation~\cite{BBCJPW93}.  Although teleportation on its
own does not appear to immediately conflict with the above intuition, we
show that, based on techniques by Vaidman~\cite{Vaidman03},
adversaries holding a large amount of entangled quantum states can
perform \emph{instantaneous nonlocal quantum computation}, which in particular implies that they
can compute any unitary operation on a state shared between them,
using only local operations and {\em one} round of classical mutual
communication. Based on this technique, we show how a coalition of
adversaries can attack and break any position-verification scheme.

Interestingly, sharing entangled quantum systems is vital for
attacking the position-verification scheme.
We show that there exist schemes that are secure in the
information-theoretic sense, if the adversary is not allowed to
pre-share or maintain entanglement. 
Furthermore, we show how to construct secure protocols for
several position-based cryptographic tasks: position-verification,
authentication, and key exchange.

This leads to an interesting open question regarding the amount of
pre-shared entanglement required to break the positioning scheme: the
case of a large amount of pre-shared states yields a complete break of
any scheme while having no pre-shared states leads to
information-theoretically secure schemes. The threshold of pre-shared
quantum systems that keeps the system secure is yet unknown.

\subsection{Related Work}\label{sec:RelatedWork}

To the best of our knowledge, quantum schemes for position-verification have first been considered by Kent in 2002 under the name
of ``quantum tagging''. Together with Munro, Spiller and Beausoleil, a
patent for an (insecure) scheme was filed for HP Labs in 2004 and
granted in 2006~\cite{KMSB06}. Their results have not appeared in the
academic literature until 2010~\cite{KMS10}. In that paper, they
describe several basic schemes and describe how to break them using
teleportation-based attacks. They propose other variations (Schemes
IV--VI in~\cite{KMS10}) not suspect to their teleportation attack and
leave their security as an open question. Our general attack shows
that these schemes are insecure as well.

Concurrent and independent of our work and the work on
quantum tagging described above, the approach of using quantum
techniques for secure position-verification was proposed by Malaney
~\cite{Mal10a, Mal10b}. However, the proposed scheme is merely claimed
secure, and no rigorous security analysis is provided. As pointed out
in~\cite{KMS10}, Malaney's schemes can also be broken by a
teleportation-based attack. 
Chandran et al.\@ have proposed and proved secure a quantum scheme for
position-verification~\cite{CFGGO10}. However, their proof implicitly assumed that the adversaries have no pre-shared entanglement; as shown
in~\cite{KMS10}, their scheme also becomes insecure without this assumption.

In a subsequent paper~\cite{LL11}, Lau and Lo use similar ideas as
in~\cite{KMS10} to show the insecurity of position-verification
schemes that are of a certain (yet rather restricted) form, which
include the schemes from~\cite{Mal10a,Mal10b} and~\cite{CFGGO10}.
Furthermore, they propose a position-verification scheme that resists
their attack, and they conjecture it secure. While these protocols 
might be secure if the
adversaries do not pre-share entanglement, our attack shows that all
of them are insecure in general.

In a recent note~\cite{Kent10}, Kent considers a different model for
position-based cryptography where the prover's position is \emph{not} his only credential, but he is assumed to additionally
share with the verifiers a classical key unknown to the adversary. In
this case, quantum key distribution can be used to expand that key ad
infinitum. This classical key stream is then used as authentication
resource.

The idea of performing ``instantaneous measurements of nonlocal
variables'' has been put forward by Vaidman~\cite{Vaidman03} and was
further investigated by Clark et al.~\cite{CCJP10}. The concept of
instantaneous nonlocal quantum computation presented here is an
extension of Vaidman's task. 
After the appearance and circulation of our work, Beigi and
K\"onig~\cite{BK11} used the technique of port-based teleportation by
Ishizaka and Hiroshima~\cite{IH08,IH09} to reduce the amount of
entanglement required to perform instantaneous nonlocal quantum
computation (from our double exponential) to exponential.

In~\cite{GLM02}, Giovannetti et al.~show how to measure the distance
between two parties by quantum cryptographic means so that only
trusted people have access to the result. This is a different kind
of problem than what we consider, and the techniques used there
are not applicable in our setting.

\subsection{Our Attack and Our Schemes in More Detail}

\paragraph{Position-Verification - A Simple Approach.}\label{sec:SimpleApproach}

Let us briefly discuss the $1$-dimensional case in which we have
two verifiers $V_0$ and $V_1$, and a prover~$P$ at position~$\pos$
that lies on the straight line between~$V_0$ and~$V_1$. Now, to verify
$P$'s position, $V_0$~sends a BB84 qubit $H^\theta\ket{x}$ to~$P$, and
$V_1$~sends the corresponding basis~$\theta$ to~$P$. The sending of
these messages is timed in such a way that $H^\theta\ket{x}$ and
$\theta$~arrive at position~$\pos$ at the same time. $P$ has to
measure the qubit in basis~$\theta$ to obtain~$x$, and immediately
send~$x$ to both~$V_0$ and~$V_1$, who verify the correctness of~$x$ and if
it has arrived ``in time''.

The intuition for this scheme is the following.
Consider a dishonest prover $\dP_0$ between $V_0$ and~$P$, and a
dishonest prover $\dP_1$ between $V_1$ and $P$.
(It is not too hard to see that additional dishonest
provers do not help.) When $\dP_0$ receives the BB84 qubit, she does
not know yet the corresponding basis $\theta$. Thus, if she measures
it immediately when she receives it, she is likely to measure it
in the wrong basis and $\dP_0$ and $\dP_1$ will not be able to
provide the correct $x$. However, if she waits until she knows the
basis~$\theta$, $\dP_0$ and $\dP_1$ will be too late in sending $x$ to~$V_1$ in time.
Similarly, if she forwards the BB84 qubit to $\dP_1$, who receives
$\theta$ before $\dP_0$ does, then $\dP_0$ and $\dP_1$ will be too
late in sending $x$ to~$V_0$.
It seems that in order to break the scheme, $\dP_0$ needs to
store the qubit until she receives the basis~$\theta$ and at the same
time send a copy of it to~$\dP_1$. But such actions are excluded by the
no-cloning principle.

\paragraph{The Attack and Instantaneous Nonlocal Quantum Computation.}

The above intuition turns out to be wrong. Using pre-shared
entanglement, $\dP_0$ and $\dP_1$ can perform quantum teleportation
which enables them (in some sense) to act coherently on the complete
state immediately upon reception. Combining this fact with the observation by
Kent  et al.~\cite{KMS10} that the Pauli-corrections resulting
from the teleportation commute with the actions of the honest prover
in the above protocol shows that colluding adversaries can perfectly
break the protocol.

Much more generally, we will show how to break \emph{any}
position-verification scheme, possibly consisting of multiple
(and interleaved) rounds. To this end, we will show how to perform \emph{instantaneous
  nonlocal quantum computation}.  In particular, we prove that any unitary operation
$U$ acting on a composite system shared between players can be
computed using only a single round of mutual classical communication.  
Based on ideas by Vaidman~\cite{Vaidman03}, the
players teleport quantum states back and forth many times in a clever
way, {\em without} awaiting the classical measurement outcomes from
the other party's teleportations.

\paragraph{Position-Verification in the No Pre-shared Entanglement (\NPE) Model.}

On the other hand, the above intuition is correct in the {\em no pre-shared entanglement} (\NPE) model,
where the adversaries are not allowed to have pre-shared entangled
quantum states prior the execution the protocol, or, more generally, prior the execution of each round of the protocol in case of multi-round schemes. 
Even though this model may be somewhat unrealistic and
artificial, analyzing protocols in this setting serves as stepping
stone to obtaining protocols which tolerate adversaries who pre-share and maintain
some {\em limited} amount of entanglement. 
But also, rigorously proving security in the restrictive (for the adversary) \NPE model is
already non-trivial and requires heavy machinery. Our proof uses the {\em strong complementary information trade-off} (CIT) due to Renes and Boileau~\cite{RB09}, 
and it guarantees that
for any strategy, the success probability of $\dP_0$ and $\dP_1$ is
bounded by approximately~$0.89$.  By repeating the above simple scheme
sequentially, we get a secure multi-round positioning scheme with
exponentially small soundness error. 
We note that when performing sequential repetitions in the \NPE model, 
the adversaries must enter each round with no entanglement; thus, they are not allowed to generate entanglement in one round, store it, and use it in the next round(s).

\paragraph{Position-based authentication and key-exchange in the \NPE Model.}
Our position-based authentication scheme is based on our position-verification
scheme. The idea is to start with a ``weak''
authentication scheme for a $1$-bit message $m$:
the verifiers and $P$ execute the secure position-verification
scheme; if $P$ wishes to authenticate $m=1$, then $P$ correctly
finishes the scheme by sending $x$ back, but if $P$ wishes to
authenticate $m=0$, $P$ sends back an ``erasure'' $\bot$
instead of the correct reply $x$ with some probability $q$ (which
needs to be carefully chosen). This authentication
scheme is weak in the sense that turning $1$ into $0$ is easy for
the adversary, but turning a $0$ into a $1$ fails with constant
probability.

The idea is to use a suitable {\em balanced} encoding of the actual
message to be authenticated, so that for any two messages, the
adversary needs to turn many $0$'s into $1$'s. Unfortunately, an
arbitrary balanced encoding is not good enough. The reason is
that we do not assume the verifiers and the honest $P$ to be
synchronized. This asynchrony allows the adversary to make use of honest $P$ who
is authenticating one index of the encoded message, in order to
authenticate another index of the modified encoded message towards the
verifiers.

Nevertheless, we show that the above approach does work for carefully chosen codes.
We show that, for instance, the bit-wise encoding which maps
$0$ into $00...0\,11...1$ and $1$ into
$11...1\,00...0$ is such a code.

Our solution borrows some ideas from~\cite{RW03, KR09, CKOR10} on authentication based on weak secrets. However, since in our setting we cannot do ”liveness” tests (to check that the verifier is alive in the protocol), the techniques from~\cite{RW03, KR09, CKOR10} do not help us directly.

Given a position-based authentication scheme, one can immediately
obtain a position-based key-exchange scheme simply by (essentially)
executing an arbitrary quantum-key-distribution scheme
(e.g.~\cite{BB84}), which assumes an authenticated classical
communication channel, and authenticate the classical communication
by means of the position-based authentication scheme.

\subsection{Organization of the paper}
In Section~\ref{sec:preliminaries}, we begin by introducing
notation, and presenting the relevant background from quantum
information theory. In Section~\ref{sec:setup}, we describe the
problem of position-verification and define our standard quantum
model, as well as the \NPE model in more detail. 
A protocol for computing any unitary operation using local operations and one round 
of classical communication is provided and analyzed in Section~\ref{sec:nonlocal}, and
in Section~\ref{sec:Impossibility} we conclude that there does not exist any
protocol for position-verification (and hence, any protocol for
position-based cryptographic tasks) in the standard quantum model.
We present our position-verification protocol in the \NPE model in
Section~\ref{sec:basicSP}. Section~\ref{sec:auth+KE} is devoted to
our position-based authentication protocol and showing how to
combine the above tools to obtain position-based key exchange.

\section{Preliminaries}\label{sec:preliminaries}

\subsection{Notation and Terminology}\label{sec:Notation}

We assume the reader to be familiar with the basic concepts of quantum
information theory and refer to~\cite{NC00} for an excellent
introduction; we merely fix some notation.

\paragraph{Qubits.}

A {\em qubit} is a quantum system $A$ with a 2-dimensional state space~$\H_A = \C^2$.
The {\em computational basis} $\set{\ket{0},\ket{1}}$ (for a qubit) is given by $\ket{0} = {1 \choose 0}$ and $\ket{1} = {0 \choose 1}$, and the {\em Hadamard basis} by $H\set{\ket{0},\ket{1}} = \set{H\ket{0},H\ket{1}}$, where $H$ denotes the 2-dimensional {\em Hadamard matrix}, which maps $\zero$ to $(\zero+\one)/\sqrt{2}$ and $\one$ to $(\zero-\one)/\sqrt{2}$.
The state space of an $n$-qubit system $A = A_1\cdots A_n$ is given by
the $2^n$-dimensional space $\H_A = (\C^2)^{\otimes n} = \C^2 \otimes \cdots \otimes \C^2$.

Since we mainly use the above two bases, we can simplify terminology and notation by identifying the computational basis $\set{\ket{0},\ket{1}}$ with the bit $0$ and the Hadamard basis $H\set{\ket{0},\ket{1}}$ with the bit $1$. Hence, when we say that an $n$-qubit state $\ket{\psi} \in (\C^2)^{\otimes n}$ is measured in basis $\theta \in \set{0,1}^n$, we mean that the state is measured qubit-wise where basis $H^{\theta_i}\set{\ket{0},\ket{1}}$ is used for the $i$-th qubit.
As a result of the measurement, the string $x \in \set{0,1}^n$ is observed with probability $|\bra{\psi}H^\theta\ket{x}|^2$, where $H^\theta = H^{\theta_1} \otimes \cdots \otimes H^{\theta_n}$ and $\ket{x} = \ket{x_1}\otimes\cdots\otimes\ket{x_n}$.

An important example of a $2$-qubit state is the {\em EPR pair}, which is given by $\ket{\Phi_{AB}} = (\zero\zero + \one\one)/\sqrt{2} \in \H_A \otimes \H_B = \C^2 \otimes \C^2$ and has the following properties: if qubit $A$ is measured in the computational basis, a uniformly random bit $x \in \set{0,1}$ is observed and qubit $B$ collapses to $\ket{x}$. Similarly, if qubit $A$ is measured in the Hadamard basis, a uniformly random bit $x \in \set{0,1}$ is observed and qubit $B$ collapses to $H\ket{x}$.

\paragraph{Density Matrices and Trace Distance.}

For any complex Hilbert space $\H$, we write $\D(\H)$ for the set of all {\em density matrices} acting on $\H$.
We measure closeness of two density matrices $\rho$ and $\sigma$ in $\D(\H)$
by their {\em trace distance}: $\dist(\rho,\sigma) := \frac12
\tr|\rho-\sigma|$. One can show that for any physical processing of two quantum states described by
$\rho$ and $\sigma$, respectively, the two states behave in an
indistinguishable way except with probability at most
$\dist(\rho,\sigma)$. Thus, informally, if $\dist(\rho,\sigma)$ is
very small, then without making a significant error, the two quantum
states can be considered equal.

\paragraph{Classical and Hybrid Systems (and States).}

Subsystem $X$ of a bipartite quantum system $XE$ is called {\em classical}, if the state of $XE$ is given by a density matrix of the form
$
\rho_{XE} = \sum_{x \in \cal X} P_X(x) \proj{x} \otimes \rho_{E}^x \, ,
$
where $\cal X$ is a finite set of cardinality $|{\cal X}| = \dim(\H_X)$, $P_X:{\cal X} \rightarrow [0,1]$ is a probability distribution, $\set{\ket{x}}_{x \in \cal X}$ is some fixed orthonormal basis of $\H_X$, and $\rho_E^x$ is a density matrix on $\H_E$ for every \mbox{$x \in \cal X$}. Such a state, called {\em hybrid} state (also known as {\em cq-}state, for {\em c}lassical and {\em q}uantum), can equivalently be understood as consisting of a {\em random variable} $X$ with distribution $P_X$ and range $\cal X$, and a system $E$ that is in state $\rho_E^x$ exactly when $X$ takes on the value $x$. This formalism naturally extends to two (or more) classical systems $X$, $Y$ etc.\@ as well as to two (or more) quantum systems.

\paragraph{Teleportation.}

The goal of teleportation is to transfer a quantum state from one location to another by only communicating classical information. Teleportation requires pre-shared entanglement among the two locations.
Specifically, to teleport a qubit $Q$  in an arbitrary (and typically unknown) state $\ket{\psi}$ from Alice to Bob, Alice performs a
Bell-measurement on $Q$ and her half of an EPR-pair, yielding a
classical measurement outcome $k \in \set{0,1,2,3}$.
Instantaneously, the other half of the corresponding EPR pair, which is held by Bob, turns into
the state $\PC_k^\dag\ket{\psi}$, where $\PC_0, \PC_1, \PC_2, \PC_3$ denote the
four Pauli-corrections $\set{\id,X,Z,XZ}$, respectively, and $\PC^\dag$
denotes the complex conjugate of the transpose of $\PC$. The classical information $k$ is then
communicated to Bob who can recover the state $\ket{\psi}$ by performing $\PC_k$ on his EPR half.
Note that the operator $\PC_k$ is Hermitian and unitary, thus $\PC^\dag_k=\PC_k$ and $\PC_k\PC^\dag_k=\id$.

\subsection{Some Quantum Information Theory}


The {\em von Neumann entropy} of a quantum state $\rho \in \D(\H)$ is given by $\Hone(\rho) := -\tr\bigl(\rho \log(\rho)\bigr)$, where here and throughout the article, $\log$ denotes the binary logarithm.
$\Hone(\rho)$ is non-negative and at most $\log(\dim(\H))$.
For a bi-partite quantum state $\rho_{AB} \in \D(\H_A \otimes \H_B)$, the {\em conditional} von Neumann entropy of $A$ given $B$ is defined as $\Hone(\rho_{AB}|B) := \Hone(\rho_{AB}) - \Hone(\rho_B)$.
In cases where the state $\rho_{AB}$ is clear from the context, we may write $\Hone(A|B)$ instead of $\Hone(\rho_{AB}|B)$.
If $X$ and $Y$ are both classical, $\Hone(X|Y)$ coincides with the classical conditional Shannon entropy.
Furthermore, in case of conditioning (partly) on a classical state, the following holds.
\begin{lemma}\label{lemma:Average}
For any tri-partite state $\rho_{ABY}$ with classical~$Y$: $\Hone(A|BY) = \sum_y P_Y(y) \Hone(\rho_{AB}^y|B)$.
\end{lemma}
Lemma~\ref{lemma:Average}
along with the concavity of $\Hone$ and Jensen's inequality
implies that for classical $Y$: $\Hone(A) \geq \Hone(A|Y) \geq 0$. 
The proof of Lemma~\ref{lemma:Average} is given in Appendix~\ref{app:technical}.

The following theorem is a generalization of the well-known Holevo bound~\cite{Holevo73} (see also~\cite{NC00}), and follows from the {\em monotonicity of mutual information}. 
Informally, it says that measuring only reduces your information. 
Formally, and tailored to the notation used here, it ensures the following.
\begin{theorem}
\label{thm:Holevo}
Let $\rho_{AB} \in \D(\H_A \otimes \H_B)$ be an arbitrary bi-partite state, and let $\rho_{AY}$ be obtained by measuring $B$ in some basis to observe (classical) $Y$. Then
$
\Hone(A|Y) \geq \Hone(A|B)
$.
\end{theorem}

For classical $X$ and $Y$, the Fano inequality~\cite{Fano61} (see also~\cite{CT91}) allows to bound the probability of correctly guessing $X$ when having access to $Y$. In the statement below and throughout the article, $\h:[0,1]\to[0,1]$ denotes the {\em binary entropy function} defined as $\h(p) = -p \log(p) - (1-p)\log(1-p)$ for $0 < p < 1$ and as $\h(p) = 0$ for $p = 0$ or $1$, and $\h^{-1}:[0,1]\to[0,\frac12]$ denotes its inverse on the branch $0 \leq p \leq \frac12$.

\begin{theorem}[Fano inequality]\label{thm:Fano}
Let $X$ and $Y$ be random variables with ranges $\cal X$ and $\cal Y$, respectively, and let $\hat{X}$ be a guess for $X$ computed solely from $Y$. Then $q:= \P[\hat{X} \!\neq\! X]$ satisfies
$$\h(q) + q \log(|{\cal X}|-1) \geq \Hone(X|Y)\ .$$
In particular, for binary $X$: $q \geq \h^{-1}(\Hone(X|Y))$.
\end{theorem}

\subsection{Strong Complementary Information Tradeoff}


The following entropic uncertainty principle, called {\em strong complementary information tradeoff} (CIT) in~\cite{RB09}
 and generalized in~\cite{BCCRR10}, is at the heart of our security proofs.
It relates the uncertainty of the measurement outcome of a system $A$ with the uncertainty of the measurement outcome when the complementary basis is used instead, and it guarantees that there can coexist at most one system $E$ that has full information on {\em both} possible outcomes. Note that by the {\em complementary} basis $\bar{\theta}$ of a basis $\theta = (\theta_1,\ldots,\theta_n) \in \set{0,1}^n$, we mean the $n$-bit string $\bar{\theta} = (\bar{\theta}_1,\ldots,\bar{\theta}_n) \in \set{0,1}^n$  with $\bar{\theta}_i \neq \theta_i$ for all $i$.

\begin{theorem}[CIT]\label{thm:CIT}
Let $\ket{\psi_{AEF}} \in \H_A \otimes \H_E \otimes \H_F$ be an arbitrary tri-partite state, where $\H_A = (\C^2)^{\otimes n}$. Let the hybrid state $\rho_{XEF}$ be obtained by measuring $A$ in basis $\theta \in \set{0,1}^n$, and let the hybrid state $\sigma_{XEF}$ be obtained by measuring $A$ (of the original state $\ket{\psi_{AEF}}$) in the complementary basis $\bar{\theta}$. Then
\begin{equation*}
\Hone(\rho_{XE}|E) + \Hone(\sigma_{XF}|F) \geq n \ .
\end{equation*}
\end{theorem}
CIT in particular implies the following (the proof is given in Appendix~\ref{app:technical}):

\begin{corollary}\label{cor:CIT}
Let $\ket{\psi_{AEF}} \in \H_A \otimes \H_E \otimes \H_F$ be an arbitrary tri-partite state, where $\H_A = (\C^2)^{\otimes n}$.
Let $\Theta$ be uniformly distributed in $\set{0,1}^n$ and let $X$ be the result of measuring $A$ in basis $\Theta$. Then
\begin{equation*}
\Hone(X|\Theta E) + \Hone(X|\Theta F) \geq n\ .
 \end{equation*}
\end{corollary}

\section{Setup and The Task of Position Verification}\label{sec:setup}
\subsection{The Security Model}\label{subsec:model}

We informally describe the model we use for the upcoming sections,
which is a quantum version of the Vanilla (standard) model introduced
in~\cite{CGMO09} (see there for a full description).  We also describe
our {\em restricted model} used for our security proof, that we call
the no pre-shared entanglement (\NPE) model. We consider entities
$V_0,\dotsc, V_k$ called {\em verifiers} and an entity $P$, the
(honest) {\em prover}.  Additionally, we consider a coalition $\dP$ of
{\em dishonest provers} (or {\em adversaries})
\smash{$\dP_0,\ldots,\dP_\ell$}. All entities can perform arbitrary
quantum (and classical) operations and can communicate quantum (and
classical) messages among them.  

For our positive results, we consider a restricted model, which prohibits entanglement between the dishonest verifiers. Specifically, the {\em \NPE model} is such that the dishonest provers enter every new round of communication, initiated by the verifiers, with no pre-shared entanglement. That is, in every round, a dishonest prover can
send an entangled quantum state only {\em after} it receives the verifier's
message, and the dishonest provers cannot maintain such an entangled state in order to use it in the next round. 
As mentioned in the introduction, considering this simple
(but possibly unrealistic) model may help us in obtaining protocols that
are secure against adversaries with {\em limited} entanglement.

For simplicity, we assume that quantum operations and communication
are noise-free; however, our results generalize to the more realistic
noisy case, assuming that the noise is low enough. We require that the
verifiers have a private and authenticated channel among themselves, which
allows them to coordinate their actions by communicating before,
during or after protocol execution. We stress however, that this
assumption does
not hold for the communication between the verifiers and~$P$: $\dP$
has full control over the destination of messages communicated between
the verifiers and~$P$ (both ways). In particular, the
verifiers do not know per-se if they are communicating with the honest
or a dishonest prover (or a coalition of dishonest provers).

The above model is extended by incorporating the notion of {\em time}
and {\em space}. Each entity is assigned an arbitrary fixed position
$\pos$ in the $d$-dimensional space~$\R^d$, and we assume that
messages to be communicated travel at fixed velocity $v$ (e.g.\ with
the speed of light), and hence the time needed for a message to travel
from one entity to another equals the Euclidean distance between the
two (assuming that~$v$ is normalized to~$1$).  This timing assumption holds for honest
and dishonest entities. We assume on the other hand that local
computations take no time.

Finally, we
assume that the verifiers have precise and synchronized clocks, so
that they can coordinate exact times for sending off messages and
can measure the exact time of a message arrival. We do not
require $P$'s clock to be precise or in sync with the verifiers.
However, we do assume that $P$ cannot be reset.

This model allows to reason as follows.
Consider a verifier $V_0$ at position $\pos_0$, who sends
a challenge $\ch_0$ to the (supposedly honest) prover claiming to be
at position $\pos$. If $V_0$ receives a reply within time
$2\d(\pos_0,\pos)$, where $\d(\cdot,\cdot)$ is the Euclidean
distance measure in $\R^d$ and thus also measures the time a message
takes from one point to the other, then $V_0$ can conclude that he
is communicating with a prover that is within distance
$\d(\pos_0,\pos)$.

We stress that in our model, the honest prover $P$ has no advantage over the dishonest provers beyond being at its position $\pos$. In particular, $P$ does not share any secret information with the verifiers, nor can he per-se authenticate his messages by any other means. 

Throughout the article, we
require that the honest prover $P$ is {\em enclosed} by the
verifiers $V_0,\ldots,V_k$ in that the prover's position $\pos \in
\R^d$ lies within the tetrahedron, i.e., convex hull, $\Hull(\pos_0,\dotsc,\pos_k) \subset \R^d$
formed by the respective positions
of the verifiers.
Note that in this work we consider only {\em stand-alone security}, i.e.,
there exists only a single execution with a single honest prover, and we
do not guarantee concurrent security.

\subsection{Secure Position Verification}
A position-verification scheme
should allow a prover $P$ at position $\pos \in \R^d$ (in
$d$-dimensional space) to convince a set of $k+1$ verifiers
$V_0,\ldots,V_k$, who are located at respective positions
$\pos_0,\dotsc,\pos_k \in \R^d$, that he is indeed at position
$\pos$. We assume that $P$ is enclosed by
$V_0,\ldots,V_k$. 
We require that the verifiers jointly accept if an honest prover $P$ is at
position $\pos$, and we require that the verifiers reject with
``high''
probability in case of a dishonest prover
that is not at position $\pos$. The latter should hold even if the
dishonest prover consist of a {\em coalition} of collaborating
dishonest provers $\dP_0,\ldots,\dP_{\ell}$ at arbitrary
positions
$\apos_0,\dotsc,\apos_\ell \in \R^d$ with $\apos_i \neq \pos$ for
all~$i$. We refer to~\cite{CGMO09} for the general formal definition
of the completeness and security of a position-verification scheme. In
this article, we mainly focus on position-verification schemes of the
following form:

\begin{definition}\label{def:1roundSP}
  A 1-round \textbf{position-verification} scheme $\SP=(\SPch,\SPresp,\SPver)$ 
  consists of the following three parts.
  A
  challenge generator $\SPch$, which outputs a list of challenges
  $(\ch_0,\ldots,\ch_k)$ and auxiliary information~$x$; a response
  algorithm $\SPresp$, which on input a list of challenges outputs a
  list of responses $(x'_0,\ldots,x'_k)$; and a verification algorithm
  $\SPver$ with $\SPver(x'_0,\ldots,x'_k,x) \in \set{0,1}$. 
  
  $\SP$ is said to have {\bf perfect completeness} if
  $\SPver(x'_0,\ldots,x'_k,x) = 1$ with probability~1 for
  $(\ch_0,\ldots,\ch_k)$ and $x$ generated by~$\SPch$ and
  $(x'_0,\ldots,x'_k)$ by~$\SPresp$ on input $(\ch_0,\ldots,\ch_k)$.
\end{definition}

The algorithms $\SPch$, $\SPresp$ and $\SPver$ are used as described in Figure~\ref{fig:1roundSP} to verify the claimed position of a prover $P$. We clarify that in order to have all the challenges arrive at $P$'s (claimed) location $\pos$ {\em at the same time}, the verifiers agree on a time $T$ and each $V_i$ sends off his challenge $\ch_i$ at time $T-\d(\pos_i,\pos)$.
As a result, all $\ch_i$'s arrive at
$P$'s position $\pos$ at time $T$. In Step~\ref{step:check}, $V_i$
receives $x'_i$ {\em in time} if $x'_i$ arrives at $V_i$'s position
$\pos_i$ at time $T+\d(\pos_i,\pos)$. Throughout the article, we use
this simplified terminology. Furthermore, we are sometimes a bit
sloppy in distinguishing a party, like $P$, from its location
$\pos$.

\begin{figure}[htb]
\begin{framed}
\small
Common input to the verifiers: their respective positions $\pos_0,\dotsc,\pos_k$, and $P$'s (claimed) position $\pos$.
\begin{enumerate}\setlength{\parskip}{0.1ex}\setcounter{enumi}{-1}
\item $V_0$ generates a list of challenges $(\ch_0,\ldots,\ch_k)$ and auxiliary information $x$ using $\SPch$,
and sends $\ch_i$ to $V_i$ for $i = 1,\ldots,k$.
\item Every $V_i$ sends $\ch_i$ to $P$ in such a way that all $\ch_i$'s
arrive at the same time at $P$'s position~$\pos$.
\item $P$ computes $(x'_0,\ldots,x'_k):= \SPresp(\ch_0,\ldots,\ch_k)$ as soon as all the $\ch_i$'s arrive,  and he sends $x'_i$ to $V_i$ for every $i$.
\item\label{step:check} The $V_i$'s  jointly accept if and only if all $V_i$'s receive  $x'_i$ in time and $\SPver(x'_0,\ldots,x'_k,x) = 1$.
\end{enumerate}
\end{framed}
 \caption{Generic 1-round position-verification scheme. }
 \label{fig:1roundSP}
\end{figure}

We stress that we allow $\SPch$, $\SPresp$ and $\SPver$ to be {\em quantum}
algorithms and $\ch_i$, $x$ and $x'_i$ to be quantum information. In our constructions, only $\ch_0$ will actually be quantum; thus, we will only require quantum
communication from $V_0$ to $P$, all other communication is
classical. Also, in our constructions, $x'_0 = \ldots = x'_k$, and $\SPver(x'_0,\ldots,x'_k,x) = 1$ exactly if $x'_i = x$ for all~$i$.

\begin{definition}
A 1-round position-verification scheme $\SP = (\SPch,\SPresp,\SPver)$ is
called {\bf $\eps$-sound} if for any position $\pos \in
\Hull(\pos_0,\dotsc,\pos_k)$, and any coalition of dishonest provers
\smash{$\dP_0,\ldots,\dP_{\ell}$} at arbitrary positions
$\apos_0,\dotsc,\apos_\ell$, all $\neq \pos$, when executing the
scheme from Figure~\ref{fig:1roundSP} the verifiers accept with
probability at most $\eps$. We write $\SP^\eps$ for such a
protocol.
\end{definition}

In order to be more realistic, we must take into consideration physical limitations of the equipment used,
such as measurement errors, computation durations, etc. Those allow a dishonest prover which
resides arbitrarily close to $P$ to appear as if she resides at $\pos$.
Thus, we assume that all the adversaries are at least $\Delta$-distanced from $\pos$, where $\Delta$ is determined
by those imperfections. For sake of simplicity, this $\Delta$ is implicit in the continuation of the paper.

A position-verification scheme can also be understood as a
(position-based) {\em identification} scheme, where the
identification is not done by means of a cryptographic key or a
password, but by means of the geographical location.

\section{Instantaneous Nonlocal Quantum Computation}\label{sec:nonlocal}

In order to analyze the (in)security of position-verification schemes,
we first address a more general task, which is interesting in its own
right: {\em instantaneous nonlocal quantum computation}\footnote{This is an extension of the task of ``instantaneous measurement of
  nonlocal variables'' introduced by Vaidman~\cite{Vaidman03}.}.
Consider the following problem, involving two parties Alice and
Bob. Alice holds $A$ and Bob holds $B$ of a tripartite system $ABE$
that is in some unknown state $\ket{\psi}$.  The goal is to apply a
known unitary transformation $U$ to $AB$, but {\em without} using any
communication, just by local operations. In general, such a task is
clearly impossible, as it violates the non-signalling principle. The
goal of instantaneous nonlocal quantum computation is to achieve
almost the above but without violating non-signalling. Specifically,
the goal is for Alice and Bob to compute, without communication, a state
$\ket{\varphi'}$ that coincides with $\ket{\varphi} = (U \otimes \I)
\ket{\psi}$ up to {\em local} and {\em qubit-wise} operations on $A$ and $B$,
where $\I$ denotes the identity on $E$. Furthermore, these local and
qubit-wise operations are determined by {\em classical} information that
Alice and Bob obtain as part of their actions. In particular, if Alice
and Bob share their classical information, which can be done with {\em one}
round of simultaneous mutual communication, then they can transform
$\ket{\varphi'}$ into $\ket{\varphi} = U \ket{\psi}$ by local
qubit-wise operations. Following ideas by Vaidman~\cite{Vaidman03}, we
show below that instantaneous nonlocal quantum computation, as
described above, is possible if Alice and Bob share sufficiently many
EPR pairs.

In the following, let $\H_A$, $\H_B$ and $\H_E$ be Hilbert spaces where the former two 
consist of $n_A$ and $n_B$ qubits respectively, i.e.,   
$\H_A = (\C^2)^{\otimes n_A}$ and $\H_B = (\C^2)^{\otimes n_B}$. Furthermore, let $U$ be a unitary matrix acting on $\H_A \otimes \H_B$. Alice holds system $A$ and Bob holds system~$B$ of an arbitrary and unknown state $\ket{\psi} \in \H_{ABE} = \H_A \otimes \H_B \otimes \H_E$. Additionally, Alice and Bob share an arbitrary but finite number of EPR pairs. 

\begin{theorem}\label{thm:local}
  For every unitary~$U$ and for every $\eps > 0$, given sufficiently
  many shared EPR pairs, there exist local operations ${\cal A}$~and~${\cal B}$, 
  acting on Alice's and Bob's respective sides, with the
  following property. For any initial state $\ket{\psi} \in \H_{ABE}$, the joint execution ${\cal
    A} \otimes {\cal B}$ transforms~$\ket{\psi}$ into~$\ket{\varphi'}$
  and provides classical outputs~$k$ to Alice and~$\ell$ to Bob,
  such that the following holds except with probability~$\eps$. 
  The state $\ket{\varphi'}$ coincides with $\ket{\varphi} = (U \otimes \I)
  \ket{\psi}$ up to local qubit-wise operations on $A$~and~$B$ that are determined 
  by~$k$ and~$\ell$.
\end{theorem}
We stress that ${\cal A}$ acts on $A$ as well as on Alice's shares of
the EPR pairs, and the corresponding holds for ${\cal
  B}$. Furthermore, being equal up to local qubit-wise operations on
$A$ and $B$ means that $\ket{\varphi} = (V^A_{k,\ell} \otimes
V^B_{k,\ell} \otimes \I) \ket{\varphi'}$, where
$\{V^A_{k,\ell}\}_{k,\ell}$ and $\{V^B_{k,\ell}\}_{k,\ell}$ are fixed
families of unitaries which act qubit-wise on $\H_A$ and $\H_B$,
respectively. In our construction, the $V^A_{k,\ell}$ and
$V^B_{k,\ell}$'s will actually be tensor products of one-qubit Pauli operators.

As an immediate consequence of Theorem~\ref{thm:local}, we get the following. 
\begin{corollary}
  For every unitary $U$ and for every $\eps > 0$, given sufficiently
  many shared EPR pairs, there exists a nonlocal operation $\cal AB$ for Alice and Bob which consists of local operations and {\em one} round of mutual communication, such that for any initial state $\ket{\psi} \in \H_{ABE}$ of the tripartite system $ABE$, the joint execution of $\cal AB$ transforms $\ket{\psi}$ into $\ket{\varphi} = (U \otimes \I)
  \ket{\psi}$, except with probability~$\eps$.
\end{corollary}

For technical reasons, we will actually prove the following extension
of Theorem~\ref{thm:local}, which is easily seen equivalent. The
difference to Theorem~\ref{thm:local} is that Alice and Bob are
additionally given classical inputs: $x$ to Alice and $y$ to Bob, and
the unitary $U$ that is to be applied to the quantum input depends on
$x$ and $y$. In the statement below, $x$ ranges over some arbitrary
but fixed finite set $\cal X$, and $y$ ranges over some arbitrary but
fixed finite set~$\cal Y$.

\begin{theorem}\label{thm:local+}
For every family $\set{U_{x,y}}$ of unitaries and for every $\eps > 0$, given sufficiently many shared EPR pairs, there exist families $\{{\cal A}_x\}$ and $\{{\cal B}_y\}$ of local operations, acting on Alice's and Bob's respective sides, with the following property. For any initial state $\ket{\psi} \in \H_{ABE}$ and for every $x \in \cal X$ and $y \in \cal Y$, the joint execution ${\cal A}_x \otimes {\cal B}_y$ transforms the state $\ket{\psi}$ into $\ket{\varphi'}$ and provides classical outputs $k$ to Alice and $\ell$ to Bob, such that the following holds except with probability $\eps$. The state $\ket{\varphi'}$ coincides with $\ket{\varphi} = (U_{x,y} \otimes \I) \ket{\psi}$ up to local qubit-wise operations on $A$ and $B$ that are determined by $k$ and $\ell$.
\end{theorem}
The solution works by teleporting states back and forth in a
clever way~\cite{Vaidman03}, but {\em without} communicating the
classical outcomes of the Bell measurements, so that only local
operations are performed. Thus, in the formal proof below, whenever we
say that a state is teleported, it should be understood in this
sense, i.e., the sender makes a Bell measurement resulting in some
classical information, and the receiver takes his shares of the EPR
pairs as the received state, but does/can not (yet) correct it.

\begin{proof}
To simplify notation, we assume that the joint state of $A$ and $B$ is pure, and thus we may ignore system $E$. However, all our arguments also hold in case the state of $A$ and $B$ is entangled with $E$. 

Next, we observe that it is sufficient to prove Theorem~\ref{thm:local+} for the case where $B$ is ``empty'', i.e., $\dim \H_B = 1$ and thus $n_B = 0$. Indeed, if this is not the case, Alice and Bob can do the following. Bob first teleports $B$ to Alice. Now, Alice holds $A'=AB$ with $n_{A'}=n_A+n_B$, and Bob's system has collapsed and thus Bob holds no quantum state anymore, only classical information. Then, they do the nonlocal computation, and in the end Alice teleports $B$ back to Bob. The modification to the state of $B$ introduced by teleporting it to Alice can be taken care of by modifying the set of unitaries $\{U_{x,y}\}$ accordingly (and making it dependent on Bob's measurement outcome, thereby extending the set $\cal Y$). Also, the modification to the state of $B$ introduced by teleporting it back to Bob does not harm the requirement of the joint state being equal to $\ket{\varphi} = U_{x,y} \ket{\psi}$ up to local qubit-wise operations. 

Hence, from now on, we may assume that $B$ is ``empty'', and we write $n$ for~$n_A$. 
Next, we describe the core of how the local operations ${\cal A}_x$ and ${\cal B}_y$ work. 
To simplify notation, we assume that ${\cal X} = \set{1,\ldots,m}$. 
Recall that Alice and Bob share (many) EPR pairs. We may assume that the EPR pairs are grouped into groups of size $n$; each such group we call a {\em teleportation channel}. Furthermore, we may assume that $m$ of these teleportation channels are labeled by the numbers $1$ up to $m$, and that another $m$ of these teleportation channels are labeled by the numbers $m+1$ up to $2m$. 

\begin{enumerate}
\item
  Alice teleports $\ket{\psi}$ to Bob, using the teleportation channel that is labeled by her input $x$.
  Let us denote her
  measurement outcome by $k_\circ \in \set{0,1,2,3}^{n}$.

\item 
For every $i \in \set{1,\ldots,m}$, Bob does the following. He applies the unitary $U_{i,y}$ to the $n$ qubits that make up his share of the EPR pairs given by the teleportation channel labeled by $i$. Then, he teleports the resulting state to Alice using the teleportation channel labeled by $m+i$. We denote the corresponding measurement outcome by $\ell_{\circ,i}$. 

\item 
Alice specifies the $n$ qubits that make up her share of the EPR pairs given by the teleportation channel labeled by $m+x$ to be the state $\ket{\varphi'}$. 
\end{enumerate}

Let us analyze the above. With probability $1/4^n$, namely if $k_\circ = 0\cdots0$, teleporting $\ket{\psi}$ to Bob leaves the state unchanged. In this case, it is easy to see that the resulting state $\ket{\varphi'}$ satisfies the required property of being identical to $\ket{\varphi} = U_{x,y} \ket{\psi}$ up to local qubit-wise operations determined by $\ell_{\circ,x}$, and thus determined by $x$ and
$\ell_\circ = (\ell_{\circ,1},\ldots,\ell_{\circ,m})$.
This proves the claim for the case where $\eps \geq 1 - 1/4^{n}$. 

We show how to reduce $\eps$. The crucial observation is that if in the above procedure $k_\circ \neq 0\cdots0$, and thus $\ket{\varphi'}$ is not necessarily identical to $\ket{\varphi}$ up to local qubit-wise operations, then 
$$
\ket{\varphi'} = V_{\ell_{\circ,x}} U_{x,y} V_{k_\circ}\ket{\psi} = V_{\ell_{\circ,x}} U_{x,y} V_{k_\circ} U_{x,y}^\dagger \ket{\varphi} \, ,
$$
where $V_{\ell_{\circ,x}}$ and $V_{k_\circ}$ are tensor products of Pauli matrices. 
Thus, setting $\ket{\psi'}:= \ket{\varphi'}$, $x' := (x,k_\circ)$ and $y' := (y,\ell_\circ)$, and 
$U'_{x',y'} := U_{x,y} V_{k_\circ} U_{x,y}^\dagger V_{\ell_{\circ,x}}$, the state~$\ket{\varphi}$ can be written as $\ket{\varphi} = U'_{x',y'} \ket{\psi'}$. 
This means, we are back to the original problem of applying a unitary, $U'_{x',y'}$, to a state, $\ket{\psi'}$, held by Alice, where the unitary depends on classical information $x'$ and $y'$, known by Alice and Bob, respectively. Thus, we can re-apply the above procedure to the new problem instance. Note that in the new problem instance, the classical inputs $x'$ and $y'$ come from larger sets than the original inputs $x$ and $y$, but the new quantum input, $\ket{\psi'}$, has the same qubit size, $n$. Therefore, re-applying the procedure will succeed with the same probability~$1/4^n$. 

As there is a constant probability of success in each round, re-applying the above procedure sufficiently many times to the resulting new problem instances guarantees that except with arbitrary small probability, the state $\ket{\varphi'}$ will be of the required form at some point
(when Alice gets $k_\circ=0\cdots0$). 
Say, this is the case at the end of the $j$-th iteration. Then, Alice stops with her part of the procedure at this point, keeps the state $\ket{\varphi'}$, and specifies $k$ to consist of $j$ and of her classical input into the $j$-th iteration (which consists of $x$ and of the $k_\circ$'s from the prior $j-1$ iterations). 
Since Bob does not learn whether an iteration is successful or not, he has to keep on re-iterating up to some bound, and in the end he specifies $\ell$ to consist of the $\ell_\circ$'s collected over all the iterations. 
The state~$\ket{\varphi'}$ equals $\ket{\varphi} = U_{x,y} \ket{\psi}$ up to local qubit-wise operations that are determined by $k$~and~$\ell$.  
\end{proof}

Doing the maths shows that the number of EPR pairs needed by Alice and
Bob in the scheme described in the proof is double exponential in
$n_A+n_B$, the qubit size of the joint quantum system. 

In recent subsequent work~\cite{BK11}, Beigi and K\"onig have used a
different kind of quantum teleportation by Ishizaka and
Hiroshima~\cite{IH08,IH09} to reduce the amount of entanglement needed
to to perform instantaneous nonlocal quantum computation to
exponential in the qubit size of the joint quantum system. It remains
an interesting open question whether such an exponentially large
amount of entanglement is necessary.

In Appendix~\ref{sec:nparties}, we explain how to perform
instantaneous nonlocal quantum computation among more than two
parties.

\section{Impossibility of Unconditional Position Verification}\label{sec:Impossibility}

In this section we show that  no position-verification scheme is secure against
a coalition of quantum adversaries in the Vanilla model.
For simplicity, we consider the one-dimensional case, with two
verifiers $V_0$ and~$V_1$, but the attack can be generalized to higher
dimensions and more verifiers.

We consider an arbitrary position-verification scheme in our model (as
specified in Section~\ref{subsec:model}). We recall that in this
model, the verifiers must base their decision solely on {\em what} the
prover replies and {\em how long} it takes him to reply, and the
honest prover has no advantage over a coalition of dishonest provers
beyond being at the claimed position\footnote{In particular, the
  prover does not share any secret information with the verifiers,
  differentiating our setting from models as described for example in
  \cite{Kent10}.}.  Such a position-verification scheme may be of the
form as specified in Figure~\ref{fig:1roundSP}, but may also be made up
of several, possibly interleaved, rounds of interaction between the
prover and the verifiers.

For the honest prover $P$, such a general scheme consists of steps
that look as follows. 
$P$ holds a local quantum register $R$, which is set to some default value at the beginning of the scheme. 
In each step, $P$ obtains a system $A$ from $V_0$ and a system $B$ from $V_1$, and $V_0$ and $V_1$ jointly keep some system $E$. Let $\ket{\psi}$ be the state of the four-partite system $ABRE$; it is determined by the scheme and by the step within the scheme we are focussing on. 
$P$ has to apply a
fixed%
\footnote{$U$ is fixed for a fixed scheme and for a fixed step within the scheme, but of course may vary for different schemes and for different steps within a scheme. }
known unitary transformation $U$ to $ABR$, and send the
(transformed) systems $A$ and $B$ back to $V_0$ and $V_1$ (and keep
$R$). Note that after the transformation, the state of $ABRE$ is given
by $\ket{\varphi} = (U \otimes \I) \ket{\psi}$, where $\I$ is the
identity acting on $\H_E$.  For technical reasons, as in
Section~\ref{sec:nonlocal}, it will be convenient to distinguish
between classical and quantum inputs, and therefore, we
let the unitary $U$ depend on classical information $x$ and $y$,
where $x$ has been sent by $V_0$ along with $A$, and $y$ has been sent
by $V_1$ along with $B$.

We show that a coalition of two dishonest provers $\dP_0$ and
$\dP_1$, where $\dP_0$ is located in between $V_0$ and $P$ and $\dP_1$
is located in between $V_1$ and $P$, can perfectly simulate the
actions of the honest prover $P$, and therefore it is impossible for
the verifiers to distinguish between an honest prover at position
$\pos$ and a coalition of dishonest provers at positions different
from $\pos$. The simulation of the dishonest provers perfectly imitates
the {\em computation} as well as the {\em timing} of an honest $P$. Since
in our model this information is what the verifiers have to base their
decision on, the general impossibility of
position-verification in our model follows.

Consider a step in the scheme as described above, but now from the
point of view of $\dP_0$ and $\dP_1$. Since $\dP_0$ is closer to
$V_0$, he will first receive $A$ and $x$; similarly, $\dP_1$ will
first receive $B$ and $y$. We specify that $\dP_1$ takes care of and
maintains the local register $R$. If the step we consider is the
{\em first} step in the scheme, the state of $ABRE$ equals
$\ket{\psi}$, as in the case of an honest $P$. In order to have an
invariant that holds for all the steps, we actually relax this
statement and merely observe that the state of $ABRE$, say
$\ket{\psi'}$, equals $\ket{\psi}$ up to local and qubit-wise
operations on the subsystem $R$, determined by classical information
$x_\circ$ and $y_\circ$, where $\dP_0$ holds $x_\circ$ and $\dP_1$
holds $y_\circ$. This invariant clearly holds for the first step in the scheme,
when $R$ is in some default state, and we will show that it also holds
for the other steps.

By Theorem~\ref{thm:local+}, it follows that without communication,
just by instantaneous local operations, $\dP_0$ and $\dP_1$ can
transform the state $\ket{\psi'}$ into a state $\ket{\varphi'}$ that
coincides with $\ket{\varphi} = (U_{x,y} \otimes \I) \ket{\psi}$ up to
local and qubit-wise transformations on $A$, $B$ and $R$, determined
by classical information $k$ (known to $\dP_0$) and $\ell$ (known to
$\dP_1$). Note that the initial state is not $\ket{\psi}$, but rather
a state of the form $\ket{\psi'} = (V_{x_\circ,y_\circ} \otimes
\I)\ket{\psi}$, where $x_\circ$ is known to $\dP_0$ and $y_\circ$ to
$\dP_1$. Thus, Theorem~\ref{thm:local+} is actually applied to the
unitary $U'_{x',y'} = U_{x,y} V_{x_\circ,y_\circ}^\dagger$, where $x'
= (x_\circ,x)$ and $y' = (y_\circ,y)$.  Given $\ket{\varphi'}$ and $k$
and $\ell$, $\dP_0$ and $\dP_1$ can exchange $k$ and $\ell$ using {\em one}
mutual round of communication and transform
$\ket{\varphi'}$ into $\ket{\varphi''}$ that coincides with
$\ket{\varphi}$ up to qubit-wise operations only on $R$, and send $A$
to $V_0$ and $B$ to $V_1$. It follows that the state of $ABE$ and the
time it took $\dP_0$ and $\dP_1$ for the computation and communication
is identical to that of an honest $P$, i.e., $\dP_0$ and $\dP_1$ have
perfectly simulated this step of the scheme. 

Finally, we see that the invariant is satisfied, when moving on to the
next step in the scheme, where $\dP_0$ and $\dP_1$ receive new $A$ and
$B$ (along with new classical $x$ and $y$) from $V_0$ and $V_1$,
respectively. Even if this new round interleaves with the previous
round in that the new $A$ and $B$ etc.\@ arrive {\em before}
$\dP_0$ and $\dP_1$ have finished exchanging (the old) $k$ and $\ell$,
it still holds that the state of $ABRE$ is as in the case of honest
$P$ up to qubit-wise operations on the subsystem $R$. It follows
that the above procedure works for all the steps and thus that $\dP_0$
and $\dP_1$ can indeed perfectly simulate honest $P$'s actions
throughout the whole scheme.

\section{Secure Position-Verification in the \NPE model}\label{sec:basicSP}

\subsection{Basic Scheme and its Analysis}

In this section we show the possibility of secure position-verification in the \NPE model.
We consider the
following basic 1-round position-verification scheme, given in Figure~\ref{fig:BasicSP}. It is based on the BB84
encoding. 

\begin{figure}[htb]
\small
\begin{protocol}
\begin{enumerate}\setlength{\parskip}{0.1ex}\setcounter{enumi}{-1}
\item $V_0$ chooses two random bits $x,\theta \in \set{0,1}$ and privately sends them to~$V_1$.
\item $V_0$ prepares the qubit $H^{\theta}\ket{x}$ and sends it to $P$, and $V_1$ sends the bit $\theta$ to $P$, so that $H^{\theta}\ket{x}$ and $\theta$ arrive at the same time at $P$.
\item When $H^{\theta}\ket{x}$ and $\theta$ arrive, $P$ measures $H^{\theta}\ket{x}$ in basis $\theta$ to observe $x' \in \set{0,1}$, and sends $x'$ to $V_0$ and $V_1$.
\item $V_0$ and $V_1$ accept if on both sides $x'$ arrives in time and $x' = x$.
\end{enumerate}
\end{protocol}
 \caption{\!\mbox{Position-verification scheme $\basicSP^\eps$ based on the BB84 encoding.}}
 \label{fig:BasicSP}
\end{figure}

We implicitly specify that parties abort if they receive any
message that is inconsistent with the protocol, for instance
(classical) messages with a wrong length, or different number of
received qubits than expected, etc.

\begin{theorem}\label{thm:BasicSP}
  The 1-round position-verification scheme
  $\basicSP^\eps$ from Figure~\ref{fig:BasicSP} is $\eps$-sound with
  $\eps = 1- \h^{-1}(\frac12)$, in the \NPE model.
\end{theorem}
Recall that $\h$ denotes the binary entropy function and $\h^{-1}$ its inverse on the branch $0 \leq p \leq \frac12$.
A numerical calculation shows that \smash{$\h^{-1}(\frac12) \geq 0.11$} and thus $\eps \leq 0.89$. A particular attack for a dishonest prover \smash{$\dP$}, sitting in-between $V_0$ and $P$, is to measure the qubit \smash{$H^{\theta}\ket{x}$} in the {\em Breidbart} basis, resulting in an acceptance probability of $\cos(\pi/8)^2 \approx 0.85$. This shows that our analysis is pretty tight.

\begin{proof}
In order to analyze the position-verification scheme it is convenient to
consider an equivalent {\em purified} version, given in
Figure~\ref{fig:BasicEPR}. The only difference between the original
and the purified scheme is the preparation of the bit~$H^{\theta}\ket{x}$.
In the purified version, it is done by preparing 
$\ket{\Phi_{AB}}=(\zero\zero + \one\one)/\sqrt{2}$ and measuring~$A$ in basis~$\theta$.
This way of preparation changes the point in time when $V_0$~measures~$A$,
and the point in time when $V_1$~learns~$x$. 
This, however, has no
influence on the view of the (dishonest or honest) prover, nor on
the joint distribution of $\theta$, $x$ and $x'$, and thus neither
on the probability that $V_0$ and $V_1$~accept. It therefore
suffices to analyze the purified version.

\begin{figure}[htb]
\small
\begin{protocol}
\begin{enumerate}\setlength{\parskip}{0.1ex}\setcounter{enumi}{-1}
\item $V_0$ and $V_1$ privately agree on a random bit $\theta \in \set{0,1}$.
\item $V_0$ prepares an EPR pair $\ket{\Phi_{AB}} \in \H_A \o \H_B$, keeps qubit~$A$ and sends~$B$ to~$P$, and $V_1$~sends the bit~$\theta$
to~$P$, so that $B$~and~$\theta$ arrive at the same time at~$P$.
\item When $B$ and $\theta$ arrive, $P$ measures $B$ in basis $\theta$ to observe $x' \in \set{0,1}$, and sends $x'$ to $V_0$ and $V_1$.
\item Only now, when $x'$ arrives, $V_0$ measures $A$ in basis
$\theta$ to observe $x$, and privately sends~$x$ to~$V_1$. 
$V_0$ and $V_1$ accept
 if on both sides $x'$~arrives in time and $x' = x$.
\end{enumerate}
\end{protocol}
 \caption{EPR version of $\basicSP^\eps$. }
 \label{fig:BasicEPR}
\end{figure}

We first consider security against two dishonest provers $\dP_0$ and
$\dP_1$, where $\dP_0$ is between $V_0$ and $P$ and $\dP_1$ is between
$V_1$ and $P$. In the end we will argue that a similar argument holds
for multiple dishonest provers on either side.

Since $V_0$ and $V_1$ do not accept if $x'$ does not arrive in time
and dishonest provers do not use pre-shared entanglement in the
\NPE-model, any potentially successful strategy of $\dP_0$ and $\dP_1$
must look as follows. As soon as $\dP_1$ receives the bit $\theta$
from $V_1$, she forwards (a copy of) it to $\dP_0$. Also, as soon as
$\dP_0$ receives the qubit $A$, she applies an arbitrary quantum
operation to the received qubit $A$ (and maybe some ancillary
  system she possesses) that maps it into a bipartite state $E_0 E_1$
(with arbitrary state space $\H_{E_0} \o \H_{E_1}$), and $\dP_0$ keeps
$E_0$ and sends $E_1$ to $\dP_1$. Then, as soon as $\dP_0$ receives
$\theta$, she applies some measurement (which may depend on $\theta$)
to $E_0$ to obtain $\hat{x}_0$, and as soon as $\dP_1$ receives $E_1$,
she applies some measurement (which may depend on $\theta$) to $E_1$ to
obtain $\hat{x}_1$, and both send $\hat{x}_0$ and $\hat{x}_1$
immediately to $V_0$ and $V_1$, respectively. We will argue that
the probability that $\hat{x}_0 = x$ {\em and} $\hat{x}_1 = x$ is
upper bounded by $\eps$ as claimed.

Let $\ket{\psi_{A\, E_0 E_1}} \in \H_A \otimes \H_{E_0} \otimes
\H_{E_1}$ be the state of the tri-partite system $A\, E_0 E_1$ after
$\dP_0$ has applied the quantum operation to the qubit $B$. Note that
in the \NPE model, the quantum operation and thus $\ket{\psi_{A\, E_0
    E_1}}$ does not depend on $\theta$.\footnote{We stress that this independency breaks down if $\dP_0$ and $\dP_1$ may start off with an entangled state, because then $\dP_1$ can act on his part of the entangled state in a $\theta$-dependent way, which makes the overall state dependent of $\theta$. }
Recall that~$x$ is obtained by measuring $A$ in either the computational (if
$\theta = 0$) or the Hadamard (if $\theta = 1$) basis. Writing $x$,
$\theta$, etc.\ as random variables $X$, $\Theta$, etc., it follows
from CIT (specifically Corollary~\ref{cor:CIT}) that $ \Hone(X|\Theta
E_0) + \Hone(X|\Theta E_1) \geq 1 \, . $ Let $Y_0$ and $Y_1$ denote
the classical information obtained by $\dP_0$ and $\dP_1$ as a result
of measuring $E_0$ and $E_1$, respectively, with bases that may depend
on~$\Theta$. By the (generalized) Holevo bound Theorem~\ref{thm:Holevo}, it
follows from the above that
$$\Hone(X|\Theta Y_0) + \Hone(X|\Theta Y_1) \geq 1\ ,$$
therefore $\Hone(X|\Theta Y_i)\geq\frac12$ for at least one $i \in \set{0,1}$.
By Fano's inequality (Theorem~\ref{thm:Fano}), we can conclude that the
corresponding error probability $q_i = \P[\hat{X}_i \!\neq\! X]$ satisfies
\(
\textstyle\h(q_i) \geq \frac12 .
\)
It thus follows that the failure probability
$$q = \P[\hat{X}_0 \!\neq\! X \vee \hat{X}_1 \!\neq\! X]
\geq \max\set{q_0,q_1} \ge \h^{-1}(\frac12)\ ,$$ and the probability
of $V_0$ and $V_1$ accepting, \smash{$\P[\hat{X}_0 \!=\! X \wedge
  \hat{X}_1 \!=\! X] = 1 - q$}, is indeed upper bounded by $\eps$ as
claimed.

It remains to argue that more than two dishonest provers in the \NPE
model cannot do any better. The reasoning is the same as
above. Namely, in order to respond in time, the dishonest provers that
are closer to $V_0$ than $P$ must map the qubit $A$---possibly
jointly---into a bipartite state $E_0E_1$ {\em without knowing
  $\theta$}, and jointly keep $E_0$ and send $E_1$ to the dishonest
provers that are ``on the other side'' of $P$ (i.e., closer to
$V_1$). Then, the reply for $V_0$ needs to be computed from $E_0$~and~$\theta$ 
(possibly jointly by the dishonest provers that are closer to
$V_0$), and the response for $V_1$ from $E_1$ and $\theta$. Thus, it
can be argued as above that the success probability is bounded by
$\eps$ as claimed.
\end{proof}

\subsection{Reducing the Soundness Error}\label{sec:SPreduce}

In order to obtain a position-verification scheme with a negligible
soundness error, we can simply repeat the 1-round scheme
$\basicSP^\eps$ from Figure~\ref{fig:BasicSP}. Repeating the scheme
$n$ times {\em in sequence}, where the verifiers launch the next
execution only after the previous one is finished, reduces the
soundness error to $\eps^n$. Recall that in the \NPE modeL defined
in Section~\ref{subsec:model}, the
adversaries must start every round without pre-shared
entanglement. Therefore, the security of the sequentually repeated
scheme follows immediately from the
security of the 1-round scheme.

\begin{corollary}
In the \NPE model, the $n$-fold sequential repetition of $\basicSP^\eps$ from
Figure~\ref{fig:BasicSP} is $\eps^n$-sound with $\eps = 1-
\h^{-1}(\frac12)$.
\end{corollary}

In terms of round complexity, a more efficient way of repeating
$\basicSP^\eps$ is by repeating it {\em in parallel}: $V_0$ sends $n$
BB84 qubits $H^{\theta_1}\ket{x_1},\ldots,H^{\theta_n}\ket{x_n}$ and
$V_1$ sends the corresponding bases $\theta_1,\ldots,\theta_n$ to $P$
so that they all arrive at the same time at $P$'s position, and $P$
needs to reply with the correct list $x_1,\ldots,x_n$ in time. This
protocol is
obviously more efficient in terms of round complexity and appears to
be the preferred solution. However, we do not have a proof for the
security of the parallel repetition of $\basicSP^\eps$.

\subsection{Position Verification in Higher Dimensions}\label{sec:3D}

The scheme $\basicSP^\eps$ can easily be extended into higher
dimensions.
The scheme for $d$ dimensions is a generalization of the scheme $\basicSP^\eps$ in
Figure~\ref{fig:BasicSP}, where the challenges of the verifiers
$V_1$, $V_2$, $\ldots$, $V_d$ form a {\em sum sharing} of the basis
$\theta$, i.e., are random $\theta_1,\theta_2, \ldots, \theta_d \in \set{0,1}$
such that their modulo-2 sum equals $\theta$.
As specified in Figure~\ref{fig:1roundSP}, the state
$H^\theta\ket{x}$ and the shares $\theta_i$ are
sent by the verifiers to $P$ such that they arrive at $P$'s
(claimed) position at the same time. $P$ can reconstruct~$\theta$ 
and measure $H^\theta\ket{x}$ in the correct basis to
obtain $x' = x$, which he sends to all the verifiers who check if
$x'$ arrives in time and equals $x$.

We can argue security by a reduction to the scheme in 1 dimension.
For the sake of concreteness, we consider 3~dimensions. For 3 dimensions, we need a set of (at least) 4
non-coplanar verifiers $V_0, \dotsc, V_3$, and the prover $P$ needs
to be located inside the tetrahedron defined by the positions of the
4 verifiers.
We consider a coalition of dishonest provers
$\dP_0,\ldots,\dP_{\ell}$ at arbitrary positions but different to
$P$. We may assume that
$\dP_0$ is closest to $V_0$. It is easy to see that there exists a
verifier~$V_j$ such that $d(\dP_0,V_j) > d(P,V_j)$. Furthermore, we
may assume that $V_j$ is not $V_0$ and thus we assume for
concreteness that it is $V_1$. We strengthen the dishonest
provers by giving them $\theta_2$ and $\theta_3$ for free from the
beginning. Since, when $\theta_2$ and $\theta_3$ are given, $\theta$
can be computed from $\theta_1$ and vice versa, we may assume that
$V_1$ actually sends $\theta$ as challenge rather than~$\theta_1$.
But now, $\theta_2$ and $\theta_3$ are just two random bits,
independent of $\theta$ and $x$, and are thus of no help to the
dishonest provers and we can safely ignore them.

As $\dP_0$ is further away from $V_1$ than $P$ is, $\dP_0$ cannot afford to store $H^\theta\ket{x}$ until he has learned $\theta$. Indeed, otherwise $V_1$ will not get a reply in time. Therefore, before she learns $\theta$, $\dP_0$ needs to apply a quantum transformation to $H^\theta\ket{x}$ with a bi-partite output and keep one part of the output, $E_0$, and send the other part, $E_1$ to $\dP_1$. Note that this quantum transformation is independent of $\theta$, as long as $\dP_0$ does not
share an entangled state with the other dishonest provers (who might know $\theta$ by now).
Then, $\hat{x}_0$ and $\hat{x}_1$, the replies that are sent to $V_0$ and $V_1$, respectively, need to be computed from $\theta$ and $E_0$ alone and from $\theta$ and $E_1$ alone. It follows from the analysis of the scheme in one dimension that the probability that both $\hat{x}_0$ and $\hat{x}_1$ coincide with $x$ is at most $\eps = 1- \h^{-1}(\frac12)$.

\begin{corollary}
The above generalization of $\basicSP^\eps$ to $d$ dimensions is
$\eps$-sound in the \NPE model with $\eps = 1- \h^{-1}(\frac12)$.
\end{corollary}

\section{Position-Based Authentication and Key-Exchange}\label{sec:auth+KE}

In this section we consider a new primitive: position-based
authentication. In contrast to position-verification, where the goal of
the verifiers is to make sure that entity $P$ is at the claimed
location $\pos$, the verifiers want to make sure that a given
message $m$ originates from an entity $P$ that is at the claimed
location $\pos$. We stress that it is not sufficient to first
execute a position-verification scheme with $P$ to ensure that $P$ is
at position $\pos$ and then have $P$ send or confirm $m$, because a
coalition of dishonest provers may do a {\em man-in-the-middle}
attack and stay passive during the execution of the positioning
scheme but modify the communicated message $m$.

Formally, in a position-based authentication scheme the prover
takes as input a message $m$ and the verifiers $V_0,\ldots,V_k$ take
as input a message $m'$ and the claimed position $\pos$ of $P$, and
we require the following security properties.
\begin{itemize}
\item{\em $\eps_c$-Completeness: } If  $m = m'$, $P$ is honest and at the claimed position $\pos$, and if there is no (coalition of) dishonest prover(s), then the verifiers jointly accept except with probability $\eps_c$.
\item{\em $\eps_s$-Soundness:} For any 
$\pos \in \Hull(\pos_0,\dotsc,\pos_k)$ and for any coalition of
dishonest provers $\dP_0,\ldots,\dP_\ell$ at locations all different
to $\pos$, if $m \neq m'$, the verifiers jointly reject except with probability $\eps_s$.
\end{itemize}

We build a position-based authentication scheme based on our
position-verification scheme. The idea is to incorporate the message to be
authenticated into the replies of the position-verification scheme. Our
construction is very generic and may also be useful for turning
other kinds of identification schemes (not necessarily
position-based schemes) into corresponding authentication schemes.
Our aim is  merely to show the existence of such a scheme; we do not
strive for optimization.
We begin by proposing a weak  position-based
authentication scheme for a 1-bit message $m$.

\subsection{Weak 1-bit authentication scheme}

Let $\SP^\eps$ be a  1-round position-verification scheme between
$k+1$ verifiers $V_0, \ldots, V_k$ and a prover $P$.
For simplicity, we assume that, like for the scheme $\basicSP^\eps$ from Section~\ref{sec:basicSP}, $x$ and $x'_0,\ldots,x'_k$ are classical, and $\SPver$ accepts if $x'_i = x$ for all $i$, and thus we understand the output of $\SPresp(\ch_0,\ldots,\ch_k)$ as a single element~$x'$ (supposed to be~$x$).
We require
$\SP^\eps$ to have perfect completeness and soundness $\eps < 1$.
We let $\bot$ be some special symbol.
We consider the weak authentication scheme given in Figure~\ref{fig:weakAuth} for a 1-bit message $m \in \set{0,1}$. We assume that $m$ has already been communicated to the verifiers and thus there is agreement among the verifiers on the message to be authenticated. The weak authentication scheme works by executing the 1-round position-verification scheme $\SP^\eps$, but letting $P$ replace his response $x'$ by $\bot$ with probability $q$, to be specified later.

\begin{figure}[htb]
\small
\begin{protocol}
\begin{enumerate}\setlength{\parskip}{0.1ex}\setcounter{enumi}{-1}
\item
$V_0$ generates $(\ch_0,\ldots,\ch_k)$ and $x$ using $\SPch$ and sends $\ch_i$
and $x$ to $V_i$ for $i=1,\ldots,k$.
\item
Every verifier $V_i$ sends $\ch_i$ to $P$ in such a way that all
$\ch_i$s arrive at the same time at $P$.
\item When the $\ch_i$s arrive, $P$ computes the authentication tag $t$ as follows and sends it back to all the verifiers.
\\ If $m = 1$ then $t := \SPresp(\ch_0,\ldots,\ch_k)$, and if $m = 0$
then $t := \bot$ with probability $q$ and \\ $t :=
\SPresp(\ch_0,\ldots,\ch_k)$ otherwise.

\item If different verifiers have received different values for $t$, or it didn't arrive in time, the verifiers abort.
\\ Otherwise, they jointly accept if  $t = x$ or both $m = 0$ and $t =
\bot$.
\end{enumerate}
\end{protocol}
 \caption{Generic position-based weak authentication scheme $\wAUTH^\eps$ for 1-bit message~$m$. }
 \label{fig:weakAuth}
\end{figure}

We analyze the success probability of an adversary authenticating a bit $m' \in \set{0,1}$. We consider the case where there is no honest prover present (we call this an {\em impersonation attack}), and the case where an honest prover is active and authenticates the bit $m \ne m'$ (we call this a {\em substitution attack}).

The following properties are easy to verify and follow from the security property of $\SP^\eps$.

\begin{lemma}\label{lemma:wAUTH}
Let $\dP$ be a coalition of dishonest provers not at the claimed
position and trying to authenticate message $m' = 1$. In case of an
impersonation attack, the verifiers accept with probability at most
$\eps$, and in case of a substitution attack (with $m = 0$), the
verifiers accept with probability at most $\delta = (1-q) + q\eps =  1-q(1-\eps) <
1$.
\end{lemma}

On the other hand, $\dP$ can obviously authenticate $m' = 0$ by
means of a substitution attack with success probability 1; however,
informally, $\dP$ has bounded success probability in authenticating
message $m' = 0$ by means of an impersonation attack unless he uses
the tag $\bot$. (This fact is used later to obtain a strong
authentication scheme.)

Let us try to extend the above in order to get a strong authentication scheme.
Based on the observation that by performing a substitution attack on $\wAUTH^\eps$,
it is easy to substitute the message bit $m = 1$ by $m' = 0$ but
non-trivial to substitute $m =0$ by $m' = 1$,
a first approach to
obtain an authentication scheme with good security might be to apply
$\wAUTH^\eps$ bit-wise to a {\em balanced encoding} of the message.
Such an encoding should ensure that for any distinct messages $m$ and $m'$, there are many positions in which the encoding of $m'$ is $1$ but the encoding of $m$ is $0$.
Unfortunately, this is not good enough. The reason is that $P$ and
the verifiers are not necessarily synchronized. For instance, assume
we encode $m = 0$ into $c = 010101...01$ and $m' = 1$ into $c' =
101010...10$, and authentication works by doing $\wAUTH^\eps$
bit-wise on all the bits of the encoded message. If $\dP$ wants to
substitute $m = 0$ by $m' = 1$ then he can simply do the following.
He tries to authenticate the first bit $1$ of $c'$ towards the
verifiers by means of an impersonation attack. If he succeeds, which
he can with constant probability, he simply authenticates the
remaining bits $01010...10$ of $c'$ by using $P$, who is happy to
authenticate all of the bits of $c = 010101...01$. Because of
this issue of $\dP$ bringing $P$ and the verifiers out of sync, we
need to be more careful about the exact encoding we use.

\subsection{Secure Position-Based Authentication Scheme}

We specify a special class of codes, which is strong enough for our purpose.


\begin{definition}\label{def:codes}
Let $c\ \in \{0,1\}^{\code}$. A vector $e \in
\{-1,0,1\}^{2{\code}}$ is called an {\bf embedding} of $c$ if by
removing all the $-1$ entries in $e$ we obtain $c$.
Furthermore, for two strings $c, c' \in \{0,1\}^{{\code}}$ we say
that $c'$ $\mathbf \lambda${\bf-dominates}
$c$ if for all embeddings $e$ and $e'$ of
$c$ and $c'$ (at least) one of the following holds: (a) the number
of positions $i\in \{1,\ldots,2{\code}\}$ for which $e_i'=1$ {\em
and} $e_i < 1$ is at least $\lambda$, or
(b) there exist a consecutive sequence of indices $I$ such that
the set $J=\Set{i\in I}{ e'_i > -1}$ has size $|J|\ge 4\lambda$
and it contains at least $\lambda$ indices  $i \in J$ with $e_i=-1$.
\end{definition}
For instance, let
$c = 00...0\,11...1$  and
$c' = 11...1\,00...0$, where the blocks of $0$'s and $1$'s are of length $N/2$.
It is not hard to see that the two codewords $N/4$-dominate each other.
However, $\tilde{c}' = 0101...01$ does not dominate
$\tilde{c} = 1010...10$, since $\tilde{c}'$ can be embedded into $\ddag
0101...01\ddag\ddag...\ddag$ and $\tilde{c}$ into
$1010...10\ddag\ddag...\ddag$, where here and later we use $\ddag$ to
represent $-1$.

\begin{definition}
A code $C$ is {\bf $\mathbf \lambda$-dominating}, if any two codewords in $C$ $\lambda$-dominate each other.
\end{definition}
We note that the requirement for $\lambda$-dominating codes can be relaxed
in various ways
to allow a greater range of codes.

Let $\wAUTH^\eps$ be the above weak authentication scheme satisfying
Lemma~\ref{lemma:wAUTH}. In order to authenticate a message
$m \in \set{0,1}^\mu$ in a strong way (with $\lambda$ a security parameter), an encoding $c$ of $m$ using a
$\lambda$-dominating code $C$ is bit-wise authenticated by means of
$\wAUTH^\eps$, and the verifiers perform statistics over the number of
$\bot$s received.
The resulting authentication scheme is given in
Figure~\ref{fig:genAuth}; as for the weak scheme, we assume that the
message $m$ has already been communicated.
\begin{figure}[htb]
\small
\begin{protocol}
\begin{enumerate}\setlength{\parskip}{0.1ex}\setcounter{enumi}{-1}
\item $\P$ and the verifiers encode $m$ into a codeword $c=(c_1,\dotsc,c_{\code})\in C$, for a $\lambda$-dominating code~$C$.
\item\label{it1:repeat} For $j = 1,\ldots,{\code}$, the following is repeated in sequence.
\begin{enumerate}[\ref{it1:repeat}.1]\setlength{\parskip}{0.1ex}
\item $P$ authenticates $c_j$ by means of $\wAUTH^\eps$. Let $t_i$ be the corresponding tag received.
\item\label{it1:check}
If $j > 4\lambda$, the verifiers compute $n_\bot(j) = |\Set{i\in\set{j-4\lambda,\ldots,j}}{c_i=0 \wedge t_i = \bot}|$.
\end{enumerate}
\item
If any  of the $\wAUTH^\eps$ executions fails, or if $n_\bot(j)  >
8q \lambda$ for some round $j>4\lambda$, the verifiers jointly
\\ reject. Otherwise, $m$ is accepted.
\end{enumerate}
\end{protocol}
 \caption{A generic position-based authentication scheme $\AUTH$. }
 \label{fig:genAuth}
\end{figure}

\begin{theorem}
The generic position-based authentication scheme $\AUTH$ (Figure~\ref{fig:genAuth}) is
$Ne^{-2q\lambda}$-complete.
\end{theorem}
\begin{proof}
  An honest prover which follows the above scheme can fail only if for
  some round $r$, $n_\bot > 8q\lambda$. Using the Chernoff
  bound~\cite{Chernoff52}, the probability of having $n_\bot >
  8q\lambda$ at a specific round $r$, is upper bounded by
  $e^{-2q\lambda}$. Using the union bound for every possible round
  $j$, we can bound the failure probability with $Ne^{-2q\lambda}$.
\end{proof}

Before we analyze the security of the authentication scheme, let us
discuss the possible attacks on it. We
treat $\dP$ as a single identity, however $\dP$ represents a
collaboration of adversaries. Similarly, we  refer
the $k+1$ verifiers as a single entity, $V$. We point
out that we do not assume that honest $P$ and $V$ have synchronized
clocks. Therefore, we allow $\dP$ to arbitrarily schedule and
interleave the ${\code}$ executions of $\wAUTH^\eps$ that $V$
performs with the ${\code}$ executions that $P$ performs. The
only restriction on the scheduling is that $P$ and $V$
perform their executions of $\wAUTH^\eps$ in the specified order.

This means that at any point in time during the attack when $P$ has executed $\wAUTH^\eps$ for the bits $c_1,\ldots,c_{j-1}$ and $V$ has executed $\wAUTH^\eps$ for the bits $c'_1,\ldots,c'_{j'-1}$ and both are momentarily inactive (at the beginning of the attack: $j = j' = 1$), \smash{$\dP$} can perform one of the following three actions. (1)  Activate $V$ to run $\wAUTH^\eps$ on $c'_{j'}$ but not activate $P$; this corresponds to an impersonation attack. (2) Activate $V$ to run $\wAUTH^\eps$ on $c'_{j'}$ and activate $P$ to run $\wAUTH^\eps$ on $c_j$; this corresponds to a substitution attack if $c_j \neq c'_{j'}$. (3) Activate $P$ to run $\wAUTH^\eps$ on $c_j$ but not activate  $V$; this corresponds to ``fast-forwarding'' $P$.
We note that \smash{$\dP$}'s choice on  which action to perform
 may be adaptive and depend on what he has seen so far.
However, since $V$ and $P$ execute $\wAUTH^\eps$ for each position within $c$ independently, information gathered from previous executions of $\wAUTH^\eps$ does not improve $\dP$'s success probability to break the next execution.

It is easy to see that any attack with its (adaptive) choices of
(1), (2) or (3)  leads to embeddings $e$ and $e'$ of $c$ and $c'$,
respectively. Indeed, start with  empty strings $e=e'=\emptyset$ and
update them as follows. For each of $\dP$'s rounds, update $e$ by $e\ddag$ and $e'$ by
$e'c'_{j'}$ if $\dP$ chooses (1), update $e$ by $ec_j$ and $e'$ by
$e'c'_{j'}$ if he chooses (2), and update $e$ by $ec_j$ and $e'$ by
$e'\ddag$ if he chooses (3). In the end, complete $e$ and $e'$ by
padding them with sufficiently many $\ddag$s to have them of length
$2{\code}$. It is clear that the obtained $e$ and $e'$ are indeed
valid embeddings of $c$ and $c'$, respectively.

\begin{theorem}\label{cor:CantAuthDom}
For any $\eps > 0$ and $0 < q < (1-\eps)/8$, the generic
position-based authentication scheme $\AUTH$
(Figure~\ref{fig:genAuth}) is
$2^{-\Omega(\lambda)}$-sound in the \NPE model.
\end{theorem}

\begin{proof}
Let $m$ and $m' \neq m$ be the messages input by $P$ and the verifiers, respectively, and let $c$ and $c'$ be their encodings. Furthermore, let $e$ and $e'$ be their embeddings, determined (as explained above) by $\dP$'s attack.
By the condition on the $\lambda$-dominating code $C$ we know that one of the
two properties (a) or (b) of Definition~\ref{def:codes} holds. If (a) holds, the number of positions
$i \in \set{1,\ldots,2{\code}}$ for which $e'_i = 1$ and $e_i \in
\set{-1,0}$ is $\lambda$. In this case, by construction of the
embeddings,  in his attack $\dP$ needs to authenticate (using
$\wAUTH^\eps$) the bit $1$ at least $\lambda$ times (by means of an
impersonation or a substitution attack). By Lemma 2, the success
probability of $\dP$ is thus at most $\delta^{\lambda}$, which is
$2^{-\Omega(\lambda)}$.
In the case where property (b) holds, there exists a consecutive sequence of indices $I$ such that the set $J=\Set{i\in I}{ e'_i > -1}$ has size $|J|\ge 4\lambda$
and contains at least $\lambda$ indices  $i \in J$ with $e_i=-1$.
For
any such index $i \in J$ with $e_i=-1$, $\dP$ needs to authenticate (using
$\wAUTH^\eps$) the bit $e'_i$ by means of an impersonation attack, while
he may use $\bot$ for (at most) a $8q$-fraction of those $i$'s.

However, by the $\eps$-soundness of $\SP^\eps$, if we require $\eps
< 1 - 8q$, the probability of $\dP$ succeeding in this attack is
exponentially small in~$\lambda$.
\end{proof}

A possible choice for a dominating code for $\mu$-bit messages is the {\em balanced repetition code} $\Cbrl^\mu$, obtained by applying the code
$\Cbrl = \{00..011..1, 11..100..0\}\subset \{0,1\}^{2\ell}$ bit-wise.

\begin{lemma}
For any $\ell$ and $\mu$, the balanced repetition code $\Cbrl^\mu$
is $\ell/4$-dominating.
\end{lemma}

\begin{proof}
Let $c,c' \in \set{0,1}^{2\ell\mu}$ be two distinct code words from $\Cbrl^\mu$, and let $e$ and $e'$ be their
respective embeddings. Note that $c$ is made up of blocks of $0$'s and $1$'s of length $\ell$.
Correspondingly, $e$ is made up of blocks of $0$'s and $1$'s of length $\ell$, with $\ddag$'s inserted at various positions. Let $I_1,\ldots,I_{2\mu}$ be the index sets that describe these 0 and 1-blocks of $e$. In other words, they satisfy: $I_j < I_{j+1}$ element-wise, $|I_j| = \ell$, and $\Set{e_i}{i \in I_j}$ equals $\set{0}$ or $\set{1}$. Furthermore, the sequence of $e_i$'s with $i \in I_1 \cup \ldots \cup I_{2\mu}$ equals $c$, and as such, for any odd $j$, one of $I_j$ and $I_{j+1}$ is a 0-block and one a 1-block.
Let $\phi : \set{1,\ldots,\mu} \to  \set{1,\ldots,2\mu}$ be the function such that $I_{\phi(k)}$ is the $k$-th 1-block in $I_1,\ldots,I_{2\mu}$.
The corresponding we can do with $c'$ and $e'$, resulting in blocks $I'_1,\ldots,I'_{2\mu}$ and function $\phi'$. For any $j$, we define  $\cl(I'_j)$ to be the smallest "interval" in $\set{1,\ldots,4\mu\ell}$ that contains $I'_j$.

For 1-blocks $I_j$ and $I'_{j'}$, we say that $I_j$ {\em overlaps} with $I'_{j'}$ if  $|I_j \cap \cl(I'_{j'})| \geq 3\ell/4$.
We make the following case distinction.

{\em Case 1:} $I_{\phi(k')}$ does not overlap with $I'_{\phi'(k')}$ for some $k'$.
If all the indices in $I_{\phi(k')} \setminus \cl(I'_{\phi'(k')})$ are larger than those in $\cl(I'_{\phi'(k')})$,
then $e'_i = 1$ for all $i \in I'_{\phi'(1)} \cup \ldots \cup I'_{\phi'(k')}$ but $e_i < 1$ for at least $\ell/4$ of these $i$'s.
A similar argument can be used when all these indices are smaller than those in $\cl(I'_{\phi'(k')})$.
If neither of the above holds, then  $e'_i = 1$ for all $i \in I'_{\phi'(k')}$ but $e_i < 1$ for at least $\ell/4$ of these $i$'s. Hence, property (a) of Definition~\ref{def:codes}  is satisfied (with parameter $\ell/4$).

{\em Case 2:} $I_{\phi(k)}$ overlaps with $I'_{\phi'(k)}$ for every $k$.
Since $c$ and $c'$ are distinct, and by the structure of the code, there must exist two subsequent 1-blocks $I_{\phi(k)}$ and $I_{\phi(k+1)}$ such that the number of 0-blocks between $I_{\phi(k)}$ and $I_{\phi(k+1)}$ is strictly smaller than the number of 0-blocks between the corresponding 1-blocks $I'_{\phi'(k)}$ and $I'_{\phi'(k+1)}$.
If there is no 0-block between $I_{\phi(k)}$ and $I_{\phi(k+1)}$ and (at least) one 0-block between $I'_{\phi'(k)}$ and $I'_{\phi'(k+1)}$ then by the assumption on the overlap, at least half of the indices $i$ in the 0-block $I'_{\phi'(k)+1}$ satisfy $e_i = \ddag$.
If there is one 0-block between $I_{\phi(k)}$ and $I_{\phi(k+1)}$ and two 0-blocks between $I'_{\phi'(k)}$ and $I'_{\phi'(k+1)}$ then at least a quarter of the indices $i \in I'_{\phi'(k)+1} \cup I'_{\phi'(k)+2}$ satisfy $e_i = \ddag$.
In both (sub)cases, property (b) of Definition~\ref{def:codes}  is satisfied (with $\lambda = \ell/4$).
\end{proof}

Plugging in the concrete secure positioning scheme from
Section~\ref{sec:3D}, we obtain a secure realization of
position-based authentication scheme in $\R^d$, in the \NPE
model.

\subsection{Position-Based Key Exchange}

The goal of a position-based key-exchange scheme is to have the
verifiers agree with honest prover $P$ at location $\pos$ on a key
$\key \in \set{0,1}^L$, in such a way that no dishonest prover has
any (non-negligible amount of) information on $\key$ beyond its
bit-length $L$, as long as he is not
located at $\pos$.%
\footnote{The length $L$ of the key may depend on the course of the scheme. In particular, an adversary may enforce it to be $0$. }
Formally,  we require the following security
properties.
\begin{itemize}
\item{\em $\eps_c$-Completeness: } If $P$ is honest and at the claimed position $\pos$, and if there is no (coalition of) dishonest prover(s), then $P$ and $V_0,\ldots,V_k$ output the same key $\key$ of positive length, except with probability~$\eps_c$.

\item{\em $\eps_s$-Security:} For any position $\pos \in
  \Hull(\pos_0,\dotsc,\pos_k)$ and for any coalition $\dP$ of
  dishonest provers at locations all different to $\pos$, the hybrid
  state $\rho_{\key E}$, consisting of the key $K$ output by the
  verifiers and the collective quantum system of $\dP$ at the end of
  the scheme, satisfies $\dist(\rho_{\key E},\rho_{\idealkey} \otimes \rho_{E}) \le \eps_s$, where \smash{$\idealkey$} is chosen independently and at random of the same bit-length as $\key$.
\end{itemize}

Note that the security properties only ensure that the {\em verifiers}
can be convinced that $\dP$ has no information on the key they obtain;
no such security is guaranteed for $P$.
 Indeed,
 $\dP$ can always honestly execute the scheme with $P$,
 acting as verifiers. Also note that the security properties do
 not provide any guarantee to the verifiers that $P$ has obtained the
 \emph{same} key that was output by the verifiers,
 in case of an active attack by $\dP$,
 but this feature can always be achieved e.g.\ with the help of a position-based authentication
 scheme by having $P$ send an authenticated hash of his key.

A position-based key-exchange scheme can easily be obtained by
taking any quantum key-distribution (QKD) scheme that requires
authenticated communication, and do the authentication by means of a
position-based authentication scheme, like the scheme from the
previous section. One subtlety to take care of is that QKD schemes
usually require {\em two-way} authentication, whereas position-based
authentication only provides authentication from the prover to the
verifiers. However, this problem can easily be resolved as follows. Whenever
the QKD scheme instructs $V_0$ (acting as Alice in the QKD scheme)
to send a message $m$ in an authenticated way to $P$ (acting as
Bob), $V_0$ sends $m$ without authentication to $P$, but in the
next step $P$ authenticates the message $m'$ he has received
(supposedly $m' = m$) toward the verifiers, who abort and output an
empty key $K$ in case the authentication fails.

Using  standard
BB84 QKD, we obtain a concrete
position-based key-exchange scheme.
The security of that scheme follows from the security of the BB84
protocol~\cite{LC99,BBBMR00-STOC, SP00,Mayers01,BHLMO05,Renner05} and of the
position-based authentication scheme.

\section{Conclusion and Open Questions} \label{sec:conclusion}
Continuing a very recent line of
research~\cite{Mal10a,Mal10b,CFGGO10,KMS10,Kent10}, we have given a general
proof that information-theoretic position-verification quantum schemes
are impossible, thereby answering an open question about the security
of schemes proposed in~\cite{KMS10} to the negative. On the positive
side, we have provided schemes secure under the assumption that
dishonest provers do not use pre-shared entanglement. Our results
naturally lead to the question: How much entanglement is needed in
order to break position-verification protocols? Can we show security
in the bounded-quantum-storage model~\cite{DFSS05} where adversaries
are limited to store, say, a linear fraction of the communicated
qubits?


\section*{Acknowledgments}
We thank Charles Bennett, Fr\'ed\'eric Dupuis and Louis Salvail for interesting
discussions.  HB would like to thank Sandu Popescu for explaining
Vaidman's scheme and pointing~\cite{CCJP10} out to him.


\small
\bibliographystyle{alpha}
\bibliography{quantum-pbc-short}
\normalsize

\appendix

\section{Proofs}\label{app:technical}

\subsection{Proof of Lemma~\ref{lemma:Average}}

In this section we prove the following lemma (Lemma~\ref{lemma:Average}):
{\it
For any tri-partite state $\rho_{ABY}$ with classical $Y$,
$$\Hone(A|BY) = \sum_y P_Y(y) \Hone(\rho_{AB}^y|B).$$
}

We first consider the case of an ``empty" $B$.
$Y$ being classical means that $\rho_{AY}$ is of the form $\rho_{AY} = \sum_y P_Y(y) \rho_{A}^y \otimes \proj{y}$.
Let us write $\lambda^y_1,\ldots,\lambda^y_n$ for the eigenvalues of $\rho_{A}^y$. Note that the eigenvalues of $\rho_{AY}$ are given by $P_Y(y) \lambda^y_i$ with $y \in \cal Y$ and $i \in \set{1,\ldots,n}$.
It follows that
\begin{align*}
\Hone(\rho_{AY}&|Y) = \Hone(\rho_{AY}) - \Hone(\rho_{Y})
= -\tr\bigl(\rho_{AY}\log(\rho_{AY})\bigr) + \tr\bigl(\rho_{Y}\log(\rho_{Y})\bigr) \\[1ex]
&= - \Big( \sum_{y,i}P_Y(y) \lambda^y_i \log\bigl(P_Y(y)\lambda^y_i\bigr) - \sum_{y} P_Y(y) \log\bigl(P_Y(y)\bigr) \Big) \\
&= - \sum_{y}P_Y(y) \sum_i \lambda^y_i \log\bigl(\lambda^y_i\bigr)
= \sum_{y} P_Y(y)  \Hone(\rho_{A}^y) \, .
\end{align*}
In general, we can conclude that
\begin{align*}
\Hone(\rho_{ABY}&|BY) = \Hone(\rho_{ABY}) - \Hone(\rho_{BY})
=  \sum_{y} P_Y(y)  \Hone(\rho_{AB}^y) -  \sum_{y} P_Y(y)  \Hone(\rho_{B}^y) \\
&=   \sum_{y} P_Y(y)  \big(\Hone(\rho_{AB}^y) -   \Hone(\rho_{B}^y) \big)
=  \sum_{y} P_Y(y)  \Hone(\rho_{AB}^y|B) \, ,
\end{align*}
which proves the claim.
\qed

\subsection{Proof of Corollary~\ref{cor:CIT}}

By Lemma~\ref{lemma:Average}, we can write
\begin{align*}\nonumber
\Hone(X|\Theta E) + \Hone(X|\Theta F) 
&=\frac{1}{2^n} \sum_{\theta} \Hone(\rho_{XE}^\theta|E) + \frac{1}{2^n} \sum_{\theta} \Hone(\rho_{XF}^\theta|F) \nonumber\\
&= \frac{1}{2^n} \sum_{\theta} \big(\Hone(\rho_{XE}^\theta|E) +\Hone(\rho_{XF}^{\bar{\theta}}|F) \big) \, .
\end{align*}
Note that $\rho_{XE}^\theta$ is obtained by measuring $A$ of $\ket{\psi_{AEF}}$ in basis $\theta$ (and ignoring $F$), and $\rho_{XF}^{\bar{\theta}}$ is obtained by measuring $A$ of $\ket{\psi_{AEF}}$ in the complementary basis $\bar{\theta}$ (and ignoring $E$).
Hence, Theorem~\ref{thm:CIT} applies and we can conclude that $\Hone(\rho_{XE}^\theta|E) +\Hone(\rho_{XF}^{\bar{\theta}}|F) \geq n$ and thus $\Hone(X|\Theta E) + \Hone(X|\Theta F) \geq n$.
\qed

\section{Instantaneous Nonlocal Quantum Computation With $N$ Parties}
\label{sec:nparties}

We generalize the above result to any $N$-party distributed computation,
by generalizing Theorem~\ref{thm:local+} to the case of $N$-parties.
We assume that some distinguished user holds the system $A$ and the information $x \in {\cal X}$, while for the rest, each user $p=1\ldots N-1$ holds the system ${B_p}$ and the classical input $y_p \in {\cal Y}_p$. Let us call the user who holds $\H_A$ Alice, and the rest of the users $\user_p$ with $p=1\ldots N-1$. Denote $\H_{all}\triangleq \H_A\otimes \H_{B_1}\otimes\cdots\otimes\H_{B_{N-1}} $.
The parties share an arbitrary and unknown state
$\ket{\psi} \in \H_{all} \otimes \H_E$, and a unitary operation $U$  defined on $\H_{all}$. The
unitary $U$ is determined by $x$ and $\{y_p\}$ out of some fixed family of unitaries.

\begin{theorem}\label{thm:local+N}
For every family $\set{U_{x,y_1,\ldots,y_{N-1}}}$ of unitaries defined
on $\H_{all}$ and for every $\eps > 0$, given sufficiently many
pairwise shared EPR pairs, there exist families $\{{\cal A}_x\}$, $\{{\cal B}^1_{y_1}\}$, $\ldots$, 
$\{{\cal B}^{N-1}_{y_{N-1}}\}$ of local operations, acting on Alice's and $\user_p$'s respective sides, with the following property. For any initial state $\ket{\psi} \in \H_{all}\otimes \H_E$ and for every $x \in \cal X$ and $y_1,\ldots,y_{N-1} \in {\cal Y}_1\times\cdots\times{\cal Y}_{N-1}$, the joint execution ${\cal A}_x \otimes {\cal B}^1_{y_1}\otimes\cdots\otimes {\cal B}^{N-1}_{y_{N-1}}$  transforms the state $\ket{\psi}$ into $\ket{\varphi'}$ and provides classical outputs~$k$ to Alice and~$\ell_p$ to~$\user_p$, such that the following holds except with probability~$\eps$. The state~$\ket{\varphi'}$ coincides with $\ket{\varphi} = (U_{x,y_1,\ldots,y_{N-1}} \otimes \I) \ket{\psi}$ up to local qubit-wise operations on systems $A$~and~${B_p}$ for $p=1\ldots N-1$, that are determined by $k$~and~$\{\ell_p\}$.
\end{theorem}

\begin{proof}
As in the two-party case, we may assume that Alice holds $\ket{\psi}$ and that for each player $\user_p$, $\dim \H_{B_p}=1$. Furthermore, we assume that the joint state of $A$ and $\{B_p\}$ is pure, and thus we may ignore system $E$. 
We prove the theorem by induction on the number of parties. As we have already proven the above for $N=2$ (and the case of $N=1$ is trivial), let us assume that
the proposition holds for $N=c$ and show it also holds for $N=c+1$.  
\begin{enumerate}
\item
Alice begins by teleporting the state $\ket{\psi}$ to $\user_1$
through teleportation channel number~$x$ she shares with~$\user_1$. Let $k_\circ \in\{0,1,2,3\}^n$ be the outcome of her measurement performed during the teleportation.

\item For every $i=1\ldots |{\cal X}|$, denote with $\ket{\varphi_i}$ the state at $\user_1$'s end of the $i^\text{th}$ teleportation channel. 
Next, for $i=1,\ldots,|{\cal X}|$, users $\user_1$ to $\user_c$ 
perform the scheme given by the induction assumption\footnote{To be more precise, the scheme is performed with the given instance $\mathbf U$, reduced to the case of $c$ classical inputs, by ``merging'' the first two inputs, i.e., $\{ U_{z_1,z_2,\ldots,z_c} \}_{z_1\in ({\cal X}\times{\cal Y}_1), z_2\in{\cal Y}_2, \ldots, z_c \in {\cal Y}_c }$.} on the input state~$\ket{\varphi_i}$ with respective classical information $((i,y_1), y_2, y_3, \ldots, y_c)$, and 
with 
$\{ U^i_{y_1,\ldots,y_{c}} := U_{x=i,y_1,\ldots,y_{c}} \}$ being the family of unitaries.
At the end of the induction step $\user_1$~holds the state~$\ket{\varphi'_i}$ and each of $\user_p$
obtains a classical output~$\ell^i_p$,\footnote{For simplifying notation, we denote by $\ell^i_1$ the classical information $k^i$ that $\user_1$ obtains when acting as the distinguished user in the scheme given by the induction assumption.} such that for every $i$ 
the state~$\ket{\varphi'_i}$ coincides with $(U_{x=i,y_1,\ldots,y_{N-1}}\otimes \I)\ket{\varphi_i}$ up to local qubit-wise operations determined by~$\{\ell^i_p\}$.

\item For every $i$, $\user_1$ teleports $\ket{\varphi'_i}$ back to Alice, using teleportation channel number~$|{{\cal X}|+i}$. Let $\ell_{\circ,i} \in \{0,1,2,3\}^n$ be the outcome of his measurement performed during the teleportation.

\item 
Alice specifies the state at her end of teleportation channel number $|{\cal X|}+x$ to be the state~$\ket{\varphi'}$. 
\end{enumerate}
Clearly, if $k_\circ=0\cdots0$ then the parties $\user_1,\ldots,\user_c$ on  teleportation channel $i=x$ perform instantaneous quantum computation of the the unitary $U_{x,y_1,\ldots,y_{c}}$ on the state $\ket{\psi}$, obtaining the state~$\ket{\varphi'_x}$ which coincides with $U_{x,y_1,\ldots,y_{c}}\ket{\psi}$ up to some local qubit-wise operations determined by their classical outputs $\ell^x_1, \ldots, \ell^x_c$, that is,
$ \ket{\varphi'_x}=W_{\ell^x_1, \ldots, \ell^x_c}U_{x,y_1,\ldots,y_{c}}\ket{\psi}$, where $W$ is a tensor product of Pauli matrices determined by their classical input. The state $\ket{\varphi'}$ obtained by Alice at the $|{\cal X}|+x$ teleportation channel coincides with $\ket{\varphi'_x}$ up to local qubit-wise operations determined by $\ell_{\circ,x}$, which proves the theorem for this case.

On the other hand, assume $k_\circ \ne 0\cdots0$, then by the induction assumption
\begin{align*}
\ket{\varphi'} &= V_{\ell_{\circ,x}} W_{\ell^x_1,\ldots, \ell^x_c} U_{x,y_1,\ldots,y_c} V_{k_\circ}\ket{\psi} \\
&= V_{\ell_{\circ,x}} W_{\ell^x_1,\ldots, \ell^x_c} U_{x,y_1,\ldots,y_c} V_{k_\circ} U_{x,y_1,\ldots,y_c}^\dagger \ket{\varphi}
\end{align*}
where $V_{\ell_{\circ,x}}$ and $V_{k_\circ}$ are tensor products of Pauli matrices, and $W_{\ell^x_1,\ldots, \ell^x_c}$ is the local qubit-wise (Pauli) operations asserted by the induction assumption.
Thus, setting $\ket{\psi'}:= \ket{\varphi'}$, $x' := (x,k_\circ)$, $y'_1 := (y_1,\ell_\circ, \ell^x_1)$ and $y'_p:=(y_p, \ell^x_p)$ for $p=2...c$, and letting 
\[U'_{x',y'_1, \ldots, y'_c} := U_{x,y_1,\ldots, y_c} V_{k_\circ} U_{x,y_1,\ldots,y_c}^\dagger W_{\ell^x_1,\ldots, \ell^x_c} V_{\ell_{\circ,x}},\] the state~$\ket{\varphi}$ can be written as $\ket{\varphi} = U'_{x',y'_1,\ldots, y'_c} \ket{\psi'}$. Again, we are back to the original problem of applying a unitary, $U'_{x',y'_1,\ldots,y'_c}$, to a state, $\ket{\psi'}$, held by Alice, where the unitary depends on classical information $x'$ and $\{y'_p\}$, known by Alice and the users $\user_p$, respectively. 
We complete the proof by recalling that the success probability per round is constant which depends only on~$\dim \H_{all}$.
Assuming sufficient number of pairwise shared EPR pairs, re-applying  the above procedure sufficiently many times to the resulting new problem instances guarantees that except with arbitrary small probability, the state $\ket{\varphi'}$ will be of the required form at some point.
\end{proof}

\end{document}